\documentclass[acmsmall,nonacm,screen]{acmart}

\newtoggle{draft}
\togglefalse{draft}

\newtoggle{extended}
\toggletrue{extended}

\usepackage{adjustbox}
\usepackage{caption}
\usepackage{subcaption}
\usepackage{graphicx}
\usepackage{wrapfig}

\usepackage{tabularx}

\usepackage{cleveref}
\newcommand{\rref}[1]{rule~\ref{#1}}
\newcommand{\Rref}[1]{Rule~\ref{#1}}

\crefformat{section}{\S#2#1#3}
\crefmultiformat{section}{\S\S#2#1#3}{ and~#2#1#3}{, #2#1#3}{, and~#2#1#3}

\crefformat{subsection}{\S#2#1#3}
\crefmultiformat{subsection}{\S\S#2#1#3}{ and~#2#1#3}{, #2#1#3}{, and~#2#1#3}

\crefformat{figure}{Fig.\@~#2#1#3}
\crefmultiformat{figure}{Figs.\@~#2#1#3}{ and~#2#1#3}{, #2#1#3}{, and~#2#1#3}

\patchcmd{\thebibliography}{\clubpenalty4000}{\clubpenalty10000}{}{}
\patchcmd{\thebibliography}{\widowpenalty4000}{\clubpenalty10000}{}{}

\usepackage{csquotes}

\usepackage{mathtools}
\usepackage{stmaryrd} 

\usepackage{mathpartir} 
\mprset{vskip=0.2ex} 

\usepackage[nolist,nohyperlinks]{acronym}
\acrodef{API}{Application Programming Interface}  \acused{API}

\newcommand{\xxrightarrow}[1]{
  \xrightarrow{\raisebox{-1.35pt}[0pt][0pt]{$\scriptstyle #1$}}
}

\newcommand{\twoheadrightarrowtail}{
  \rightarrowtail\mathrel{\mkern-14mu}\rightarrow
}

\usepackage{xspace}
\newcommand*{\eg}{e.g.\@\xspace}
\newcommand*{\ie}{i.e.\@\xspace}

\usepackage{calc}
\usepackage{relsize}
\newcommand{\posc}[1]{_{\makebox[\widthof{\texttt{x}}]{%
  \smaller[3]%
  \StrLen{#1}[\mystringlen]%
  \ifthenelse{\mystringlen < 2}{}{\hspace{0.4em}}%
  \ensuremath{#1}%
}}\ifmmode\else\@\fi}
\newcommand{\id}[1]{\texttt{#1}}
\newcommand{\posid}[2]{\id{#1}_{#2}}
\newcommand{\posidc}[2]{\id{#1}\ensuremath{\posc{#2}}} 
\newcommand{\qhole}[1]{\ensuremath{\mathcolor{violet}{{\scalebox{0.9}{$\llparenthesis$}}} #1 \mathcolor{violet}{{\scalebox{1}{$\rrparenthesis$}}}}}
\newcommand{\qholec}[2][1]{\makebox[#1\fontcharwd\font`x]{\qhole{#2}}} 

\newcommand{\rhole}[2][1]{\qholec[#1]{\small{\rightarrow}#2}}
\newcommand{\rholec}[2][1]{\qholec[#1]{\tiny{\rightarrow}\scriptsize{#2}}}

\newcommand{\ty}[1]{\mathbf{#1}}
\newcommand{\tyINT}{\ty{int}}

\newcommand{\UV}[1]{\text{?$#1$}}     

\usepackage{listings}
\usepackage{lstautogobble}

\makeatletter
\let\old@lstKV@SwitchCases\lstKV@SwitchCases
\def\lstKV@SwitchCases#1#2#3{}
\makeatother

\usepackage{lstlinebgrd}

\makeatletter
\let\lstKV@SwitchCases\old@lstKV@SwitchCases

\lst@Key{numbers}{none}{%
    \def\lst@PlaceNumber{\lst@linebgrd}%
    \lstKV@SwitchCases{#1}%
    {none:\\%
     left:\def\lst@PlaceNumber{\llap{\normalfont
                \lst@numberstyle{\thelstnumber}\kern\lst@numbersep}\lst@linebgrd}\\%
     right:\def\lst@PlaceNumber{\rlap{\normalfont
                \kern\linewidth \kern\lst@numbersep
                \lst@numberstyle{\thelstnumber}}\lst@linebgrd}%
    }{\PackageError{Listings}{Numbers #1 unknown}\@ehc}}
\makeatother

\usepackage{tolcolors}
\makeatletter
\def\mathcolor#1#{\@mathcolor{#1}}
\def\@mathcolor#1#2#3{%
  \protect\leavevmode
  \begingroup
    \color#1{#2}#3%
  \endgroup
}
\makeatother

\definecolor{codekeywords}{rgb}{0.2,0.2,0.5}
\definecolor{codegreen}{rgb}{0.0,0.4,0.2}
\definecolor{codegray}{rgb}{0.5,0.5,0.5}
\definecolor{codepurple}{rgb}{0.58,0,0.82}
\colorlet{codebghl}{lightgray!70!white}

\usepackage[most]{tcolorbox}
\tcbset{
  on line, 
  boxsep=3pt, left=0pt,right=0pt,top=0pt,bottom=0pt,
  colframe=white,colback=codebghl,  
  highlight math style={enhanced}
}
\DeclareRobustCommand{\hl}[1]{\tcbox{#1}}

\NewEnviron{HugeAngles}{%
  \vspace{-0.2\baselineskip} 
  \tikz[baseline=0pt]
    \draw[line width=1pt] 
     node[append after command={
      (\tikzlastnode.north west) -- ($(\tikzlastnode.west)+.03*(\tikzlastnode.west)!1!90:(\tikzlastnode.north west)$) -- (\tikzlastnode.south west)
      (\tikzlastnode.north east) -- ($(\tikzlastnode.east)+.03*(\tikzlastnode.east)!1!270:(\tikzlastnode.north east)$) -- (\tikzlastnode.south east)
      }] { \BODY };
 \vspace{-0.2\baselineskip} 
}%
{}

\preto\equation{\setcounter{equation}{0}}
\makeatletter
\pretocmd\start@gather{\setcounter{equation}{0}}{}{}
\makeatother

\newcommand{\plhdr}{\mathord{\color{black!33}\bullet}}%

\lstdefinestyle{mystyle}{
    backgroundcolor=\color{white},
    commentstyle=\color{codegreen},
    keywordstyle=\bfseries\color{codekeywords},
    numberstyle=\tiny\color{codegray},
    stringstyle=\color{codepurple},
    basicstyle=%
      \ttfamily%
      \mdseries%
      \lst@ifdisplaystyle\scriptsize\fi,  
    numbers=left,      
    xleftmargin=2em,   
}

\lstdefinelanguage{MiniMod}{
    keywords={mod, include, import, val, var, lam, null},
    sensitive=true,           
    comment=[l]{//},          
    comment=[s]{/*}{*/},      
    string=[b]"               
}%
\newcommand{\MiniMod}{\lstinline[language=MiniMod]}
\newcommand{\Java}{\lstinline[language=Java]}
\newcommand{\inlineJava}{\lstinline[language=Java]}

\usepackage[noend]{algpseudocode} 
\algrenewcommand\algorithmicindent{1.0em}%

\algnewcommand{\LeftComment}[1]{\(\triangleright\) #1}
\makeatletter
\newlength{\trianglerightwidth}
\algnewcommand{\LineComment}[1]{\Statex \hskip\ALG@thistlm $\triangleright$ #1}
\algnewcommand{\LineCommentCont}[1]{\Statex \hskip\ALG@thistlm%
  \parbox[t]{\dimexpr\linewidth-\ALG@thistlm}{\hangindent=\trianglerightwidth \hangafter=1 \strut$\triangleright$ #1\strut}}
\makeatother

\algdef{SE}[SUBALG]{Indent}{EndIndent}{}{\algorithmicend\ }%
\algtext*{Indent}
\algtext*{EndIndent}

\algnewcommand{\Where}{\textbf{where}}

\algrenewcommand\algorithmicfunction{\textbf{fun}}

\algnewcommand{\FnName}[1]{\textsc{#1}}
\algnewcommand{\FnType}[1]{\textsc{#1}}

\makeatletter
\algnewcommand{\SAlgText}[1]{%
  \parbox[t]{\dimexpr\linewidth-\ALG@thistlm}{\strut#1\strut}}
\makeatother

\makeatletter
\algnewcommand{\AlgText}[1]{\Statex \hskip\ALG@thistlm%
  \parbox[t]{\dimexpr\linewidth-\ALG@thistlm}{\strut#1\strut}}
\makeatother

\algdef{SE}[STRUCT]{Struct}{EndStruct}[1]{\textbf{struct} \textsc{#1}}{\textbf{end struct}}

\algblockdefx[SWHILE]{SWhile}{EndSWhile}[1]
  {\algorithmicwhile\ \textsc{#1}}
  {\algorithmicend}
\makeatletter
\ifthenelse{\equal{\ALG@noend}{t}}{\algtext*{EndSWhile}}{}
\makeatother

\algblockdefx[SREPEAT]{SRepeat}{EndSRepeat}[1]
  {\algorithmicrepeat\ \textsc{#1}}
  {\algorithmicend}
\makeatletter
\ifthenelse{\equal{\ALG@noend}{t}}{\algtext*{EndSRepeat}}{}
\makeatother

\newcommand{\svar}[1]{\mathit{#1}}                                
\newcommand{\tuple}[1]{(#1)}
\newcommand{\set}[1]{\left\{\mkern1mu #1 \mkern1mu\right\}}

\newcommand{\qBase}[4]{
  #2 \xvisible[#1]{#3} #4}

\newcommand{\orderSyntax}{\raisebox{1.5pt}[0pt][0pt]{\ensuremath{\scriptscriptstyle o}}}

\newcommand{\cTrue}{\mathsf{true}}
\newcommand{\cExists}[1]{\exists#1.\:}
\newcommand{\cConj}{\mathrel{*}}
\newcommand{\cEq}{\mathbin{\scriptstyle\smash{\stackrel{?}{=}}}}    
\newcommand{\cUser}[1]{\mathsf{#1}}                                 
\newcommand{\cNew}{\nabla}
\newcommand{\cEdge}[1][]{\scopeedget[#1]}
\newcommand{\cQuery}[5][]{\smash{
    \qBase{#1}{#2}{#3}{#4} \mathrel{\mapsto} #5
}}

\newcommand{\cEmp}{\mathsf{emp}}
\newcommand{\cFalse}{\mathsf{false}}
\newcommand{\cSingle}[2]{\mathsf{single}(#1, #2)}
\newcommand{\cForall}[3]{\forall#1\: \mathsf{in}\: #2.\: #3}
\newcommand{\cDataOf}[2]{\mathsf{dataOf}(#1, #2)}

\newcommand{\reclos}[1]{\ensuremath{#1^\ast}}
\newcommand{\reopt}[1]{\ensuremath{#1{}^?}}

\newcommand{\lblOrdLexVAR}{\lblVAR < \lblIMPORT < \lblLEX}
\newcommand{\lblOrdLexMOD}{\lblMOD < \lblLEX}

\newcommand{\rTurnstile}{\, \leftarrow\, } 
\newcommand{\hoPred}[2]{\mathsf{#1}_{#2}}                     

\newcommand{\stBigAnd}{{\raisebox{-.4ex}{\scalebox{2}{$\ast$}}}}

\newcommand{\opsRSState}[4]{
  \big\langle\, #1\: \big|\: #2\: \big|\: #3\: \big|\: #4\, \big\rangle
}

\usepackage{fontawesome5}

\usepackage{scopegraphs}
\newcommand{\SG}{\mathcal{G}}

\newcommand{\lbl}[1]{\mathsf{#1}}

\newcommand{\lblLEX}{\lbl{LEX}}
\newcommand{\lblIMP}{\lbl{IMP}}

\newcommand{\lblMOD}{\lbl{MOD}}

\newcommand{\lblVAR}{\lbl{VAR}}

\newcommand{\lblIMPORT}{\lbl{IMP}}

\newcommand{\javaTotalTests}{2528}

\newcommand{\chocopyTotalTests}{196}

\newcommand{\fgjTotalTests}{49}

\newcommand{\allTotalTests}{2773}
\newcommand{\allSuccessfulTests}{2575}

\newcommand{\javaOriginalTestSet}{4105}
\newcommand{\chocopyOriginalTestSet}{308}
\newcommand{\fgjOriginalTestSet}{52}

\usepackage{ournotes}

\theoremstyle{acmplain}
\newtheorem*{theorem*}{Theorem}

\theoremstyle{remark}
\newtheorem*{remark}{Remark}

\iftoggle{extended}{
  \title{Language-Parametric Reference Synthesis (Extended)}
  \titlenote{This is an extended version of the paper \href{https://doi.org/10.1145/3720481}{published in PACMPL OOPSLA'25} \cite{Pelsmaeker_OOPSLA25}.}
}{
  \title{Language-Parametric Reference Synthesis}
}

\author{Daniel A. A. Pelsmaeker}
\orcid{0000-0003-0196-0567}
\affiliation{
  \department{Software Technology}
  \institution{Delft University of Technology}
  \city{Delft}
  \country{Netherlands}
}
\email{d.a.a.pelsmaeker@tudelft.nl}
\authornote{Both authors contributed equally to the paper.}

\author{Aron Zwaan}
\orcid{0000-0002-1818-4245}
\affiliation{
  \department{Software Technology}
  \institution{Delft University of Technology}
  \city{Delft}
  \country{Netherlands}
}
\email{a.s.zwaan@tudelft.nl}
\iftoggle{extended}{
  \authornotemark[2]
}{
  \authornotemark[1]
}

\author{Casper Bach}
\orcid{0000-0003-0622-7639}
\affiliation{
  \department{Software Technology}
  \institution{Delft University of Technology}
  \city{Delft}
  \country{Netherlands}
}
\email{c.b.poulsen@tudelft.nl}
\authornote{Author's current affiliation: University Of Southern Denmark, Odense, Denmark, casperbach@imada.sdu.dk}

\author{Arjan J. Mooij}
\orcid{0009-0005-9566-7696}
\affiliation{
  \institution{TNO-ESI}
  \city{Eindhoven}
  \country{Netherlands}
}
\affiliation{
  \institution{Z\"urich University of Applied Sciences}
  \city{Winterthur}
  \country{Switzerland}
}
\email{arjan.mooij@tno.nl}



\iftoggle{extended}{
  \setcopyright{cc}
  \setcctype{by}
  \copyrightyear{2025}
}{
  \setcopyright{acmlicensed}
  \acmJournal{PACMPL}
  \acmYear{2025} \acmVolume{9} \acmNumber{OOPSLA1} \acmArticle{123} \acmMonth{4}\acmDOI{10.1145/3720481}
}

\begin{document}

\begin{abstract}
  Modern Integrated Development Environments (IDEs) offer automated refactorings to aid programmers in developing and maintaining software.
  However, implementing sound automated refactorings is challenging, as refactorings may inadvertently introduce name-binding errors or cause references to resolve to incorrect declarations.
  To address these issues, previous work by Sch\"afer et al.\@ proposed replacing concrete references with \emph{locked references} to separate binding preservation from transformation.
  Locked references vacuously resolve to a specific declaration, and after transformation must be replaced with concrete references that also resolve to that declaration.
  Synthesizing these references requires a faithful inverse of the name lookup functions of the underlying language.

  Manually implementing such inverse lookup functions is challenging due to the complex name-binding features in modern programming languages.
  Instead, we propose to automatically derive this function from type system specifications written in the Statix meta-DSL.
  To guide the synthesis of qualified references we use \emph{scope graphs}, which represent the binding structure of a program, to infer their names and discover their syntactic structure.

  We evaluate our approach by synthesizing concrete references for locked references in \javaTotalTests{} Java, \chocopyTotalTests{} ChocoPy, and \fgjTotalTests{} Featherweight Generic Java test programs.
  Our approach yields a principled language-parametric method for synthesizing references.
\end{abstract}

\maketitle

\pagebreak[4]  

\section{Introduction}%
\label{sec:introduction}

\begin{displayquote}[{\citet[\S 3]{EkmanSV08}}][{.}]
  \textins*{P}reserving bindings is at the heart of any refactoring\\ that moves, creates, or duplicates code
\end{displayquote}

\noindent
As software projects evolve, their code is frequently refactored to improve their structure and maintainability.
Refactoring often involves copying or moving code from one code unit (such as a class, module, or trait) to another, in a way that preserves the program's behavior.
A crucial aspect of behavior-preserving transformations is name binding preservation, to ensure references in refactored code resolve to the same distinct declarations as before.
While behavior preservation also needs control- and data flow analysis, name binding preservation can be achieved using only the static semantic analysis of the program.
However, due to the sophisticated name binding features found in many modern programming languages, preserving the name resolution semantics of code across transformations is generally challenging.

To illustrate the complexity of reasoning about advanced name binding features, consider the Java program shown in~\cref{fig:java-rename-example-before}.
If we rename the field \Java|x| (line~2) to \Java|y|, Java's static semantics would cause the reference to \Java|y| on line 7 (in method \Java|foo|) to resolve to the newly renamed field \Java|y| on line~2, rather than the intended declaration of \Java|y| on line 5.
This undesired change would alter the name binding structure of the program.
To prevent this, the reference to \Java|y| on line 7 should be \emph{qualified} as \Java|Outer.this.y|, as shown in the refactored example in \cref{fig:java-rename-example-after}.

Transformations that require name binding preservation are common across many refactorings, such as those from Fowler's catalog~\cite{Fowler99}.
Yet, manually refactoring code is time-consuming and error-prone.
Consequently, many modern Integrated Development Environments (IDEs) provide automated refactorings such as \textsc{Rename}, \textsc{Inline/Extract Method}, and \textsc{Pull Up/Push Down}~\cite{Fowler99}, which attempt to automatically \emph{requalify} references to maintain the program's binding structure.

However, even popular IDEs for mainstream languages struggle to implement sound refactorings.
For example, \citet{EkmanSV08} identify several bugs in Eclipse~3.4 where automated refactorings inadvertently altered the program's binding structure.
These errors arise from the difficulty of accurately determining which references need to be fixed and computing the correct requalifications.
Not only references in the modified code, but references \emph{throughout the entire code base} may require requalification.
Ensuring both \emph{soundness} (preserve name bindings) and \emph{completeness} (finding all possible requalifications) is particularly difficult.

\Citeauthor{EkmanSV08} conclude that these challenges are ``not related to the core ingredients of the implemented refactoring, \textins*{but} inherent to the complexity of name binding rules in mainstream languages.''
As a result, existing research on the sound requalification of references is often language-specific, focusing on mainstream languages like Java~\cite{SchaferTST12}.
Implementing sound automated refactorings for other languages, like Domain-Specific Languages (DSLs) with small language developer teams, can require a prohibitively high effort.
As such, a more principled and language-parametric approach to guarantee name binding preservation is needed.

\begin{figure}[t]
  \input{fig/0100-java-rename-example}
  \caption{
    \textsc{Rename} refactoring of a small Java program, where renaming the field \Java|x| to \Java|y| on line 2 requires the reference \Java|x| on line 7 to be appropriately qualified.
  }%
  \label{fig:java-rename-example}
\end{figure}

\begin{figure}[t]
  \input{fig/0100-java-rename-example-intermediate-steps}
  \caption{
    Intermediate steps for performing the \textsc{Rename} refactoring from~\cref{fig:java-rename-example} using locked references.
    After locking the relevant reference $\id{y}$ to declaration $\posid{y}{2}$ (\hyperref[fig:java-rename-example-intermediate-steps-before]{a}) and performing the transformation (\hyperref[fig:java-rename-example-intermediate-steps-after]{b}), our approach would synthesize a solution for the locked reference and obtain~\cref{fig:java-rename-example-after}.
  }%
  \label{fig:java-rename-example-intermediate-steps}
\end{figure}

\subsection{Locked References}%
\label{subsec:language-parametric-locked-references}

\Citet{EkmanSV08} observe that many bugs in automated refactorings could ``be avoided if a set of carefully crafted building blocks were available to refactoring developers.''
One such building block is \emph{locked references}%
\footnote{
  Terminology introduced by~\citet{SchaferTST12}.
  Also referred to as ``bound names''~\citep{ecoop09refactoring}, ``locked names''~\citep{SchaferMOOPSLA2010}, and ``locked bindings''~\citep{SchaferTST12}.
  We use ``locked references'' throughout this paper.
},
proposed by Sch\"afer et al.\@ in previous work~\cite{SchaferEM08,SchaferMOOPSLA2010}.
A locked reference is an abstract reference that continues to refer to the same unique declaration even if code is moved or the declaration is renamed.
This ensures that transformations cannot cause such a reference to accidentally capture a different declaration.

The following diagram summarizes program transformation with locked references:

\[
  \mathcal{P}
  \xrightarrow{\textsf{lock}}
  \mathcal{P}^{}_{\text{locked}}
  \xrightarrow{\textsf{transform}}
  \mathcal{P}'_{\text{locked}}
  \xrightarrow{\textsf{unlock}}
  \mathcal{P}'
\]%
\vspace{0em}

\noindent
Before refactoring, we first `$\mathsf{lock}$' each relevant concrete reference by replacing it with a locked reference pointing to the original declaration.
In~\cref{fig:java-rename-example-intermediate-steps-before} we replace the concrete reference $\id{y}$ (line~7) with a locked reference $\rhole[5]{\posid{y}{2}}$ to the declaration $\posid{y}{2}$ on line~5.\footnote{
  Our syntax for locked references $\rhole[5]{d}$ is inspired by the syntax \citet{OmarVHAH17} use for holes.
} (We use subscript indices to distinguish different occurrences of the same name, but the indices are not part of the syntax.)

Next, we `$\mathsf{transform}$' the program as required for the refactoring, renaming declarations and moving code.
In~\cref{fig:java-rename-example-intermediate-steps-after}, the declaration on line 2 is renamed to $\id{y}$.
Finally, we `$\mathsf{unlock}$' each locked reference in the program by replacing it with a \emph{synthesized} concrete reference that unambiguously resolves to the intended declaration.
In this case, unlocking replaces the locked reference with \inlineJava|Outer.this.y|, maintaining the name binding semantics of the program (see~\cref{fig:java-rename-example-after}).

Every step in this pipeline gives rise to challenges, but in this paper we focus on the key challenge of synthesizing concrete references when unlocking locked references.
The program should remain well-typed and synthesized references should resolve to their intended declarations.
Separating name binding preservation from the transformation guarantees that refactorings preserve name bindings, and also makes it easier to implement refactorings.

There are numerous potential applications of reference synthesis.
In the line of work by~\citet{SchaferEM08,SchaferTST12}, it can be applied to implement sound (editor) refactorings.
Furthermore, it provides a powerful transformation tool for implementing sound transformations of DSL programs or performing large-scale codebase transformations aiming to improve the overall code quality.
However, one can also envision \emph{user-extensible refactoring tools} (such as presented by Li and Thompson~\citep{LiT12-19}) or \emph{transformation languages} (such as IntelliJ's structural search and replace) that need to preserve name bindings.
Finally, it could be used to build an editor service that suggests fixes for type errors (\eg, \textsc{quick fix} in Eclipse).

\subsection{Language-Parametric Reference Synthesis}%
\label{subsec:language-parametric-reference-synthesis}

Following \citet{SchaferEM08,SchaferTST12}, a reference synthesis function can be thought of as the right inverse of a reference resolution function.
That is, if $\mathit{QRef}$ is the set of qualified references, $\mathit{Decl}$ is the set of uniquely identified declarations, and the function $\mathsf{resolve}_p : \mathit{QRef} \rightharpoonup \mathit{Decl}$ resolves a reference at some location $p$ in the program, then locked reference synthesis should be a function $\mathsf{synthesize}_p : \mathit{Decl} \rightharpoonup \mathit{QRef}$ such that $\mathsf{resolve}_p(\mathsf{synthesize}_p(d)) = d$ for any $p$.\footnote{\Citeauthor{SchaferEM08} use ``lookup'' instead of ``resolve'' and ``access'' instead of ``synthesize'', but the idea is the same.}
The $\mathsf{lock}$ function uses $\mathsf{resolve}$ to obtain the target declaration when replacing a concrete reference with a locked reference, and conversely, $\mathsf{unlock}$ uses the $\mathsf{synthesize}$ function to replace the locked reference with a concrete reference.

The reference synthesis pipeline shown above is conceptually \emph{language-parametric}.
However, as discussed before, implementing correct reference synthesizers manually is error-prone and time-consuming.
In this paper, we present a \emph{language-parametric} approach to derive the $\mathsf{synthesize}$ function automatically from declarative type system specifications, letting language designers generically synthesize valid concrete references for programs with multiple locked references.
The goal of reference synthesis is to find for each locked reference in a program a valid concrete reference that resolves to the intended declaration.

We derive the $\mathsf{synthesize}$ function from \emph{only} a declarative specification of the language's static semantics, specified using the \emph{Statix} specification language~\cite{AntwerpenPRV18,RouvoetAPKV20}.
Statix allows syntax-directed typing rules to be specified, uses \emph{scope graphs} as a declarative model of static name binding and name resolution~\cite{NeronTVW15}, and generates executable type checkers by interpreting specifications as constraint programs.
In our implementation of $\mathsf{synthesize}$, we reinterpret the \emph{name resolution queries} from the specification to infer syntax for concrete references that resolve to a particular target declaration, guaranteeing name binding preservation.
We reuse the Statix solver (see~\cref{sec:scope-graphs-and-statix}) to validate that the syntax we infer is \emph{sound} with respect to the typing rules and represents a reference that resolves to the intended declaration.

This paper makes the following technical contributions:

\begin{itemize}

  \item
    We present a language-parametric implementation of the $\mathsf{synthesize}$ function (see~\cref{sec:operational-semantics} and~\cref{sec:heuristics}).
    This function automatically synthesizes concrete references that resolve to the specified declaration, and that are sound with respect to a type system specification written in Statix.

  \item 
    We evaluate our implementation on \allTotalTests{} test programs of Java, ChocoPy, and Featherweight Generic Java (\cref{sec:evaluation}).
    Our results demonstrate that our approach applies to mainstream languages with complex name binding semantics without modifying their typing rules.

\end{itemize}

\noindent
We first (\cref{sec:scope-graphs-and-statix}) introduce scope graphs and Statix.
Next, in~\cref{sec:reference-synthesis-by-example} we illustrate our reference synthesis algorithm.
Then, in~\cref{sec:operational-semantics} we give an operational semantics of our $\mathsf{synthesize}$ function's implementation, and the heuristics we apply in~\cref{sec:heuristics}.
We evaluate our implementation in~\cref{sec:evaluation}, and discuss related work in~\cref{sec:related-work}.
We conclude in~\cref{sec:conclusion}.

\section{Scope Graphs and Statix}%
\label{sec:scope-graphs-and-statix}

This paper presents a \emph{language-parametric} approach to synthesizing concrete references.
To this end, we build on existing work: (1) \emph{scope graphs} as a language-parametric model of name binding,
and (2) \emph{Statix} as a uniform representation of typing rules.

In this section we first describe what scope graphs are (\cref{subsec:scope-graphs}), and how they let us resolve references via graph search (\cref{subsec:scope-graph-queries}).
Then we provide a high-level introduction to the Statix language (\cref{subsec:statix-rules}) and its constraint solver (\cref{subsec:statix-constraint-solver}).

\subsection{Scope Graphs}%
\label{subsec:scope-graphs}
The example program in~\cref{fig:scope-graph-example-a} contains two declarations named~$\id{x}$, namely $\posid{x}{1}$ on line 1 and $\posid{x}{2}$ on line 3, and a named reference $\id{x}$ on line 7.\footnote{
  We use \emph{subscript indices} to distinguish different \emph{occurrences} of a particular name.
  For example, $\posid{x}{2}$ uniquely identifies one of the declarations named $\id{x}$.
}
The question is: does~$\id{x}$ refer to declaration~$\posid{x}{1}$ or~$\posid{x}{2}$?
Either can be true, depending on the semantics of the programming language.

Scope graphs~\cite{NeronTVW15,RouvoetAPKV20,AntwerpenPRV18,AntwerpenNTVW16,ZwaanA23} offer a uniform model for name resolution that supports sophisticated name binding patterns in programming languages.
As their name suggests, scope graphs model the scoping structure of programs as graphs.
Such graphs let us model both nested and recursive scoping structures,
and name resolution policies as graph search queries.

To illustrate this, consider the program and its scope graph in~\cref{fig:scope-graph-example}.
The nodes in the graph represent scopes: $s_0$ represents the \emph{global scope}, while~$s_{\id{A}}$ and~$s_{\id{B}}$ represent the scopes of modules~$\id{A}$ and~$\id{B}$, respectively.
Scopes $s_{\id{x}1}$, $s_{\id{x}2}$, and $s_{\id{y}}$ represent named declarations.
A scope $s$ may have data~$d$ associated with it, written as $s \mapsto d$, such as the name of the module that they correspond to or the name and type of their declaration.
As we shall see later, associating scopes with names lets us define \emph{name resolution queries} that resolve module names to their corresponding scopes.

Edges between scopes represent \emph{reachability} relations.
Queries can follow these edges to reach other scopes and declarations.
The two module scopes are reachable from the global scope via $\lblMOD$ edges, representing the fact that $\id{A}$ and $\id{B}$ are modules declared in the global scope $s_0$.
The module scopes $s_{\id{A}}$ and $s_{\id{B}}$ are also lexical children of the global scope, so each is connected via a $\lblLEX$ edge to $s_0$.
$\lblVAR$ edges connect scopes to declared names.
Due to the wildcard \MiniMod{import A::*} on line 6 (which imports all members of the module \id{A}) module scope $s_{\id{B}}$ is connected to module scope $s_{\id{A}}$ via an $\lblIMPORT$ edge, making all declarations in module \id{A} reachable from module \id{B}.

\begin{figure}
  \subcaptionbox{
    Program.
    \label{fig:scope-graph-example-a}
  }[0.25\textwidth]{
    \input{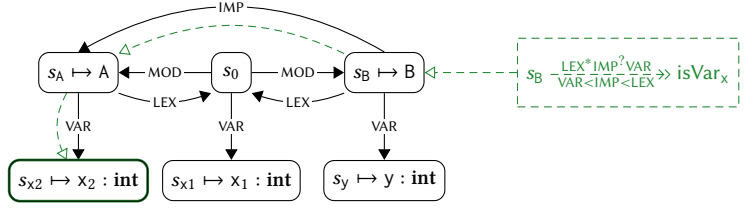}
  }%
  \hfill
  \subcaptionbox{
    Scope graph for the program in~\cref{fig:scope-graph-example-a}.
    \label{fig:scope-graph-example-b}
  }[0.74\textwidth]{
    \centering
\begin{tikzpicture}[
  scopegraph,
  node distance = 2.5em and 3.5em
]
  \node[scope] (s0) {$s_0$};
  \node[scope, below = of s0] (sx1) {$s_{\id{x}1} \mapsto \posid{x}{1} : \tyINT$};
  \draw (s0) edge[lbl=$\lblVAR$] (sx1);

  \node[scope, left = of s0] (sA) {$s_{\id{A}} \mapsto \id{A}$};
  \draw (s0) edge[lbl=$\lblMOD$] (sA);
  \draw (sA.south east) edge[lbl=$\lblLEX$, bend right=20] (s0.south west);
  \node[scope, below = of sA] (sAx2) {$s_{\id{x}2} \mapsto \posid{x}{2} : \tyINT$};
  \draw (sA) edge[lbl=$\lblVAR$] (sAx2);

  \node[scope, right = of s0] (sB) {$s_{\id{B}} \mapsto \id{B}$};
  \draw (s0) edge[lbl=$\lblMOD$] (sB);
  \draw (sB.south west) edge[lbl=$\lblLEX$, bend left=20] (s0.south east);
  \node[scope, below = of sB] (sBy) {$s_{\id{y}} \mapsto \id{y} : \tyINT$};
  \draw (sB) edge[lbl=$\lblVAR$] (sBy);

  \draw (sB.north) edge[bend right, lbl=$\lblIMPORT$] (sA.north);

  \begin{scope}[color=colorblind-bright-3]
    \node[ref, right = of sB] (qx) {
      $\qBase{\lblOrdLexVAR}{
        s_\id{B}
      }{
        \reclos{\lblLEX}\reopt{\lblIMPORT}\lblVAR
      }{
        \hoPred{isVar}{\id{x}}
      }$};
    \draw (qx) edge[ref] (sB);


    
    \draw (sA) edge[ref, bend right] (sAx2);
    \node (sAx2-green) [scope, fit=(sAx2), inner sep=0, line width=1.0pt] {}; 
    \node (sAx2-green) [scope, color=black, fit=(sAx2), inner sep=0] {};      

    \draw (sB) edge[ref, bend right=25] (sA);

  \end{scope}
\end{tikzpicture}
  }%
  \caption{
    An example LM~\cite{NeronTVW15} program and its scope graph, where boxes and arrows represent scopes and reachability relations between scopes, respectively.
    The {\color{colorblind-bright-3}dashed box} represents a query and the {\color{colorblind-bright-3}dashed arrows} its resolution path to $\posid{x}{2}$, the second occurrence of a declaration named $\id{x}$.
  }%
  \label{fig:scope-graph-example}
\end{figure}

\subsection{Scope Graph Queries}%
\label{subsec:scope-graph-queries}

We define name resolution as \emph{queries} in scope graphs.
Resolving a query entails finding all paths from this source scope to matching declarations.
To explain the syntax of a name resolution query, we take the query shown in the {\color{colorblind-bright-3}dashed box} on the right of the graph in \cref{fig:scope-graph-example-b}:

\begin{equation*}
  \qBase{
    \lblOrdLexVAR
  }{
    s_{\id{B}}
  }{
    \reclos{\lblLEX}\reopt{\lblIMPORT}\lblVAR
  }{
    \hoPred{isVar}{\id{x}}
  }
\end{equation*}

\noindent
Here, $s_{\id{B}}$ is the initial scope of the graph search, and $\hoPred{isVar}{\id{x}}$ is a filter that ensures only declarations with name \id{x} are selected.
The regular expression $\reclos{\lblLEX}\reopt{\lblIMPORT}\lblVAR$ is a \emph{reachability policy} declaring which declarations are reachable; \ie, those declarations we can reach by following a sequence of labeled edges that match the regular expression.
The path ordering $\lblOrdLexVAR$ is a \emph{visibility policy} used to disambiguate which reachable names are visible, \ie, to model shadowing.
For example, both $s_{\id{x}1}$ and $s_{\id{x}2}$ are reachable in~\cref{fig:scope-graph-example-b}.
However, the order prefers $\lblIMPORT$ edges over $\lblLEX$ edges, so the only valid path through the graph is the path to $s_{\id{x}2}$.

\subsection{Statix Rules and Constraints}%
\label{subsec:statix-rules}
In classical typing rules, terms are typed relative to one or more \emph{typing contexts}~\cite{Pierce2002}, or typed via \emph{symbol tables}~\cite{au72} or \emph{class tables}~\cite{IgarashiPW01}.
Following existing work~\cite{AntwerpenPRV18,PoulsenRTKV18,RouvoetAPKV20}, we can define typing rules in a similar style, but with terms typed relative to one or more scopes in a scope graph instead.
The constraint language Statix~\cite{AntwerpenPRV18,RouvoetAPKV20,AntwerpenV21} lets us declare such inference rules using a syntax inspired by logic programming.
Type system specifications written in Statix have a declarative interpretation, specifying a class of well-typed programs.
Alternatively, specifications can be used operationally to type check programs by constructing a scope graph and resolving references by traversing it.
This subsection highlights the main features of Statix rules and constraints.
For a more detailed breakdown of the syntax, we refer to the discussion in~\cref{subsec:statix-syntax} and the work of~\citet{RouvoetAPKV20}.

\begin{figure}[t]
  \input{fig/0200-alg-minimod-rules}
  \caption{
    A subset of the typing rules of LM, a toy language from~\cite{NeronTVW15} used for the examples in this paper.
  }%
  \label{fig:alg-minimod-rules}%
  \vspace{-0.5\baselineskip} 
\end{figure}

The rules in~\cref{fig:alg-minimod-rules} show a representative subset of the Statix rules we derived for LM, a toy language from~\cite{NeronTVW15} used throughout this paper.
The figure declares rules for five different typing relations: $\cUser{typeOfExpr}$, $\cUser{memberOk}$, $\cUser{modOk}$, $\cUser{importOk}$, and $\cUser{scopeOfMod}$.
Each rule has a conclusion on the left of an arrow ($\leftarrow$), and a premise given by one or more constraints on the right.
For example, the \rref{eq:t-add} states that the expression $\svar{e_1} + \svar{e_2}$ has type $\tyINT$ in scope $\svar{s}$, if both $\svar{e_1}$ and $\svar{e_2}$ have type $\tyINT$ under the same scope $s$.
\Rref{eq:t-qref}~is a more complex example, where a qualified module access expression $a.x$ has type $T$ when the $a$ resolves to a module
(asserted by the predicate constraint $\cUser{scopeOfMod}(\svar{s}, \svar{a}, \svar{s_m})$),
and $x$ resolves from that module to a declaration of type $T$
(asserted by the query constraint \smash{$\qBase{}{\svar{s_m}}{\lblVAR}{\hoPred{isVar}{\svar{x}}}$}).

The rules in~\cref{fig:alg-minimod-rules} do not mention the underlying scope graph explicitly.
Instead, premises of rules assert requirements on the scope graph structure, such as the existence of scopes with associated data ($\cNew{\svar{s_x}} \mapsto \svar{x} : \svar{T}$) and edges ($\svar{s} \cEdge[\lblVAR] \svar{s_x}$) and the ability to resolve query constraints.
A program is well-typed when a \emph{minimal} scope graph exists that satisfies each such assertion and query.
Minimality implies that the scope graph only has the scopes and edges asserted by the rules of a program: no extraneous edges or scopes.
There exists a solver for Statix constraints that computes this minimal scope graph, which we discuss in the next section.

\subsection{Statix Constraint Solver}%
\label{subsec:statix-constraint-solver}
Following~\citet{RouvoetAPKV20}, the operational semantics of Statix is given by
a constraint solver that soundly constructs and queries scope graphs, and uses unification to solve equality constraints.
The reference synthesis approach we illustrate in \cref{sec:reference-synthesis-by-example} is sound by construction because it builds on this operational semantics.
We defer a deeper discussion of the operational semantics to \cref{sec:operational-semantics}.

The Statix solver will solve as many constraints as possible, yielding either a state with no unsolved constraints (\ie, the program type-checks), a state that derives $\mathsf{false}$ (\ie, the program does not type-check), or a \emph{stuck} state, where the solver does not have enough information to solve the remaining constraints.
There are two reasons why constraints get stuck:
either (1) it is not sufficiently instantiated,
or (2) it is a query constraint which is not yet guaranteed to yield a \emph{stable answer}.
For (1), the solver will only expand a predicate such as $\cUser{typeOfExpr}(x, y, z)$ once $x$,~$y$, and $z$ are sufficiently instantiated such that only a single rule matches.
Similarly, it will only run and solve a query constraint once its \emph{source scope} (\eg, $s$ in \smash{$\qBase{\lblOrdLexVAR}{s}{\reclos{\lblLEX}\reopt{\lblIMPORT}\lblVAR}{\hoPred{isVar}{\svar{x}}}$}) and \emph{data well-formedness predicate} (\eg, $\hoPred{isVar}{\svar{x}}$) are ground.
In case (2), a query gets stuck when it needs to run but another unsolved constraint might add a scope graph edge that could invalidate the query.
The Statix solver implements guards that detect these cases~\citep{RouvoetAPKV20}.

Our $\mathsf{synthesize}$ function runs the Statix solver on a program with holes, where each hole is represented by a free unification variable that maps to a target scope representing the hole's intended target declaration.
The unification variables cause the Statix solver to get stuck on predicate and query constraints directly related to the holes.
Once a stuck state is reached, our reference synthesis approach extends the usual operational semantics of Statix with the ability to use the typing rules of a language to refine the holes of the program, and use the Statix solver to verify the solution.
Once the term of a hole becomes ground and all constraints in the state have been solved, we have successfully synthesized a concrete reference.

We will illustrate how the Statix solver, scope graph, and typing rules are used in our reference synthesis algorithm in \hyperref[sec:reference-synthesis-by-example]{the next section}, and discuss the operational semantics of Statix and our extension in more detail in~\cref{sec:operational-semantics}.

\section{Reference Synthesis by Example}%
\label{sec:reference-synthesis-by-example}

As the name suggests, \emph{reference synthesis} is used to synthesize a concrete (qualified) \emph{reference} to a given declaration.
References in many languages take the shape $x_1.x_2. \cdots .x_n$, modulo syntax.
Here, the first name $x_1$ is resolved from the place in the program where the reference occurs, and subsequent names $x_i$ are resolved relative to wherever the previous qualifiers $x_1. \cdots .x_{i-1}$ led.
The final name $x_n$ leads to the target declaration.

This informal definition of a reference encompasses many syntactic constructs that we intuitively recognize as (qualified) references across languages, for example \Java|Person.this.name| in Java, \lstinline[language=C]|std::option::Option| in Rust, and \lstinline[language=Cobol]|ID| \lstinline[language=Cobol]|IN| \lstinline[language=Cobol]|CUSTOMER| \lstinline[language=Cobol]|IN| \lstinline[language=Cobol]|LAST-TRANSACTION| in Cobol. 
On the other hand, according to our definition, syntax like \Java|List<String>| in Java does not constitute a reference: it is a parameterized type, akin to how a method call \Java|foo(x, y)| would also not be considered a reference.
We give a more precise definition of a reference in~\cref{sec:operational-semantics}.

\[
  \color{gray!75!black}
  \mathcal{P} =
  X\overline{[r]} \xrightarrow{\textsf{lock}}
  X\overline{[\rhole[5]{d}]} \xrightarrow{\textsf{transform}}
  {\color{black}
  X'\overline{[\rhole[5]{d}]} \xrightarrow{\textsf{unlock}}                 
  X'\overline{[r']}
  }
  = \mathcal{P}'
\]%

\noindent
The above diagram reiterates program transformation with locked references.
Initially, we have a program $\mathcal{P}$, which can be represented as a context\footnote{We can regard a `context' as a zipper-like structure over an AST.} $X$ with references $\overline{r}$ occurring in it.
First, we $\mathsf{lock}$ relevant references.
That is, we replace concrete references with locked references that durably remember which declaration they point to, regardless of where the locked reference occurs in the program (\smash{$X\overline{[\rhole[5]{d}]}$}).
Then we $\mathsf{transform}$ the program, possibly moving code around (\smash{$X'\overline{[\rhole[5]{d}]}$}).
Finally, we $\mathsf{unlock}$ the locked references: synthesizing their concrete references ($\overline{r'}$) and plugging them into the program to yield the concrete transformed program (\smash{$\mathcal{P}' = X'\overline{[r']}$}).

In this section, we use a simple program with a locked reference, shown in~\cref{fig:minimod-locked-reference-example}, to illustrate the semantics of our $\mathsf{synthesize}$ function.
The idea is to model locked references as holes given by unification variables, and strategically apply typing rules to infer a substitution for each hole.
The strategy used to apply typing rules must guarantee that inferred substitutions correspond to references that resolve as intended.

\begin{figure}[t]
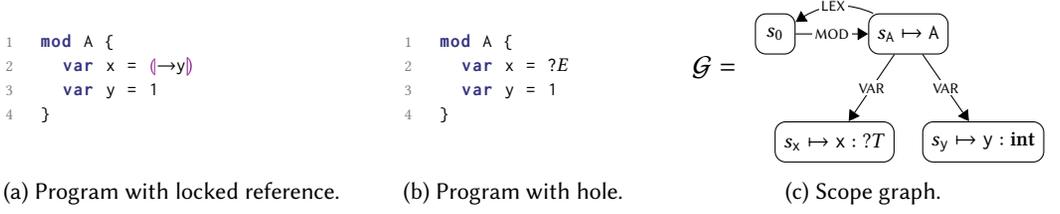

  \centering
  \subcaptionbox{
    Program with locked reference.
    \label{fig:minimod-locked-reference-example-a}
  }[0.33\textwidth]{
    \input{fig/0300-reference-synthesis-example-program-a}
  }%
  \hfill
  \subcaptionbox{
    Program with hole.
    \label{fig:minimod-locked-reference-example-b}
  }[0.22\textwidth]{
    \input{fig/0300-reference-synthesis-example-program-b}
  }%
  \hfill
  \subcaptionbox{
    Scope graph.
    \label{fig:minimod-locked-reference-example-c}
  }{
    \begin{minipage}[r]{0.06\textwidth}
      $\mathcal{G} =$
    \end{minipage}%
    \begin{minipage}[l]{0.28\textwidth}
      \vspace{1em}

\begin{tikzpicture}[scopegraph,node distance=2em and 2.5em]
  \node[scope] (s_0) {$s_0$};
  \node[scope,right=2.75em of s_0] (s_A) {$s_{\id{A}} \mapsto \id{A}$};
  \node[scope,below left=2.5em and -1em of s_A] (s_x) {$s_{\id{x}} \mapsto \id{x} : \UV{T}$};
  \node[scope,below right=2.5em and -1em of s_A] (s_y) {$s_{\id{y}} \mapsto \id{y} : \tyINT$};

  \draw (s_0) edge[lbl=$\lblMOD$] (s_A);
  \draw (s_A) edge[lbl=$\lblVAR$] (s_x);
  \draw (s_A) edge[lbl=$\lblVAR$] (s_y);
  \draw (s_A) edge[lbl=$\lblLEX$, bend right] (s_0);
\end{tikzpicture}%
    \end{minipage}%
  }
  \caption{Small example program and its scope graph\hspace{0.02\textwidth}}%
  \label{fig:minimod-locked-reference-example}
\end{figure}

\subsection{Initial Constraint Solving}
Consider the example in~\cref{fig:minimod-locked-reference-example-a}, where for the locked reference $\rhole[4]{\id{y}}$ we want to synthesize a concrete reference that must resolve to variable $\id{y}$'s declaration scope in the underlying scope graph.
Valid concrete references that our approach could yield include $\id{y}$ and $\id{A.y}$.

The first step of our approach is to replace each locked reference in the program by a hole, represented by a fresh unification variable, shown in~\cref{fig:minimod-locked-reference-example-b}.
Then we use the original language's static semantic rules and run the Statix solver on the input program.
The presence of holes causes the solver to yield a stuck state, where the solver neither has enough information to solve all constraints nor can derive $\mathsf{false}$, due to the free unification variables.

In our example, the Statix solver will recursively expand predicates and solve scope graph constraints to yield a solver state, shown below, with the inferred scope graph $\mathcal{G}$ (see~\cref{fig:minimod-locked-reference-example-c}) and a single constraint that is stuck because Statix cannot infer which rule to apply to expand the predicate.
Additionally, in the state we record that the unification variable $\UV{E}$ is associated with hole~$h$, and that hole~$h$ should become a concrete reference that resolves to the scope $s_{\id{y}}$, which is the scope associated with the locked reference's target declaration $\id{y}$ in $\mathcal{G}$.

\begin{figure}[H]
  \begin{HugeAngles}
    \begin{minipage}[c]{0.03\textwidth}
      \hyperref[fig:minimod-locked-reference-example-c]{$\mathcal{G}$}
    \end{minipage}
    \vline
    \begin{minipage}[c]{0.49\textwidth}
      \begin{gather*}
        \cUser{typeOfExpr}(s_{\id{A}}, \hl{$\UV{E}$}, \UV{T})
      \end{gather*}
    \end{minipage}
    \vline
    \begin{minipage}[c]{0.12\textwidth}
      \begin{gather*}
        \ \hl{$\UV{E} \mapsto h$}
      \end{gather*}
    \end{minipage}
    \vline
    \begin{minipage}[c]{0.26\textwidth}
      \begin{gather*}
        \hl{$ h \mapsto (s_{\id{y}}, \UV{E})$}
      \end{gather*}
    \end{minipage}
  \end{HugeAngles}
\end{figure}

\noindent
The state has the form \smash{$\opsRSState{\mathcal{G}}{\overline{C}}{U}{H}$}.
Following \citet{RouvoetAPKV20}, the solver state is given by the (partially constructed) scope graph $\mathcal{G}$ and the set of yet unsolved constraints $\overline{C}$.
To support reference synthesis, we have augmented the solver state with $U$ and $H$.
$U$ is a partial function that maps free unification variables to hole identifiers where applicable.
$H$ maps each hole identifier in a program to a pair $(s_t, t)$ of the current target scope $s_t$ and the term $t$ synthesized for the hole so far.
Note that $s_t$ changes as we synthesize qualifiers for the concrete reference.

\subsection{Forking States}
By default, the Statix solver only expands a constraint when there is exactly one possible expansion, and otherwise the constraint gets stuck.
To support reference synthesis we augment the solver to allow solver states to be \emph{forked}.
This way we can obtain the solver state for each possible expansion of a constraint.
This way we can speculatively apply each possible expansion of a constraint, obtaining a solver state.
Forked solver states that fail are discarded, but those that can be successfully solved represent programs for which we have synthesized a valid reference.

For the stuck state discussed above, we can speculatively expand the stuck $\cUser{typeOfExpr}$ predicate constraint and fork the state for each of the $\cUser{typeOfExpr}$ rules shown in~\cref{fig:alg-minimod-rules}, yielding different states.
However, some of those rules will not lead to well-formed references, and therefore we only expand to rules that could yield a reference.
For the simple LM language shown in \cref{fig:alg-minimod-rules}, the relevant rules are \ref{eq:t-var} and \ref{eq:t-qref}.
We fork the solver state, and discuss how each of these two rules leads to a synthesized reference.

\subsection{Expanding Query Constraints}
Applying the \rref{eq:t-var} and solving the constraints as far as possible yields the first forked solver state with one stuck query constraint:

\begin{figure}[H]
  \begin{HugeAngles}
    \begin{minipage}[c]{0.03\textwidth}
      \hyperref[fig:minimod-locked-reference-example-c]{$\mathcal{G}$}
    \end{minipage}
    \vline
    \begin{minipage}[c]{0.49\textwidth}
      \begin{gather*}
        \hl{$
          \vphantom{()}  
          \cQuery[\lblOrdLexVAR]{\svar{s}_{\id{A}}}{
            \reclos{\lblLEX}\reopt{\lblIMPORT}\lblVAR
          }{
            \hoPred{isVar}{\UV{x}}
          }{
            \set{\tuple{\_, \UV{x} : \UV{T}}}
          }
      $}
      \end{gather*}
    \end{minipage}
    \vline
    \begin{minipage}[c]{0.12\textwidth}
      \begin{gather*}
        \ \UV{x} \mapsto h
      \end{gather*}
    \end{minipage}
    \vline
    \begin{minipage}[c]{0.26\textwidth}
      \begin{gather*}
        \ h \mapsto \big( s_{\id{y}}, \hl{$\UV{x}$} \big)
      \end{gather*}
    \end{minipage}
  \end{HugeAngles}
\end{figure}

\noindent
As the free variable $\UV{x}$ in the stuck constraint is related to hole $h$, we can infer that the query is also related to hole $h$: attempting to solve the constraint could be fruitful for synthesizing a reference to the target scope $s_{\id{y}}$.
Therefore, reference synthesis searches for valid scope graph paths from $s_{\id{A}}$ to the intended target the scope $s_{\id{y}}$, while respecting the reachability regex and visibility ordering of the query.
There is exactly one such path, namely the one-step path traversing the $\lblVAR$ edge from $s_{\id{A}}$ to $s_{\id{y}}$.
This implies that $\UV{x} = \id{y}$ and $\UV{T} = \tyINT$.
This solves all constraints and results in the mapping $h \mapsto (s_{\id{y}}, \id{y})$.
As all constraints have been solved and the term for the hole is ground, the unqualified reference $\id{y}$ is returned as a solution.

\subsection{Qualified Reference}
In the second forked solver state, applying the \rref{eq:t-qref} and solving constraints yields the following solver state with a stuck predicate constraint and a stuck query constraint, where the term for the hole is representing a possible \emph{qualified} reference $\UV{a}\id{.}\UV{x}$:

\begin{figure}[H]
  \begin{HugeAngles}
    \begin{minipage}[c]{0.03\textwidth}
      \hyperref[fig:minimod-locked-reference-example-c]{$\mathcal{G}$}
    \end{minipage}
    \vline
    \begin{minipage}[c]{0.49\textwidth}
      \begin{gather*}
        \hl{$
          \cUser{scopeOfMod}(\svar{s}_{\id{A}}, \UV{a}, \UV{s_m})
        $}
        \\
        \hl{$
          \vphantom{()}  
          \cQuery{\UV{s_m}}{
            \lblVAR
          }{
            \hoPred{isVar}{\UV{x}}
          }{
            \set{\tuple{\_, \UV{x} : \UV{T}}}
          }
      $}
      \end{gather*}
    \end{minipage}
    \vline
    \begin{minipage}[c]{0.12\textwidth}
      \begin{gather*}
        \ \UV{a} \mapsto h
        \\
        \ \UV{x} \mapsto h
      \end{gather*}
    \end{minipage}
    \vline
    \begin{minipage}[c]{0.26\textwidth}
      \begin{gather*}
        \ h \mapsto \big( s_{\id{y}}, \hl{$\UV{a}\id{.}\UV{x}$} \big)
      \end{gather*}
    \end{minipage}
  \end{HugeAngles}
\end{figure}

\noindent
Next, we expand the $\cUser{scopeOfMod}$ predicate.
One of the possible expansions (using~\rref{eq:s-mod}) yields the following state:

\begin{figure}[H]
  \begin{HugeAngles}
    \begin{minipage}[c]{0.03\textwidth}
      \hyperref[fig:minimod-locked-reference-example-c]{$\mathcal{G}$}
    \end{minipage}
    \vline
    \begin{minipage}[c]{0.465\textwidth}
      \begin{gather}
        \hl{$
          \vphantom{()}  
          \cQuery[\lblOrdLexMOD]{\svar{s}_{\id{A}}}{
            \reclos{\lblLEX}\lblMOD
          }{
            \hoPred{isMod}{\UV{a}}
          }{
            \set{\tuple{\UV{p}, \UV{a}}}
          }
        $}
        \\ 
        \hl{$
          \UV{s_m} \cEq \mathsf{tgt}(\UV{p})
        $}
        \\
        \cQuery{\UV{s_m}}{
          \lblVAR
        }{
          \hoPred{isVar}{\UV{x}}
        }{
          \set{\tuple{\_, \UV{x} : \UV{T}}}
        }
      \end{gather}
    \end{minipage}%
    \quad  
    \vline
    \begin{minipage}[c]{0.12\textwidth}
      \begin{gather*}
        \ \UV{a} \mapsto h
        \\
        \ \UV{x} \mapsto h
      \end{gather*}
    \end{minipage}
    \vline
    \begin{minipage}[c]{0.26\textwidth}
      \begin{gather*}
        h \mapsto \big( s_{\id{y}}, \UV{a}\id{.}\UV{x} \big)
      \end{gather*}
    \end{minipage}
  \end{HugeAngles}
\end{figure}

\noindent
The solver state has two stuck query constraints: the module query (1) and the variable query~(3).
Both query constraints have unification variables that relate to the hole $h$, so we cannot know which of these query should resolve to the target scope $s_{\id{y}}$.
Therefore, we fork the solver state again: one branch where we attempt to expand query (1) and one where we attempt to expand query (3).
Only the fork that attempts to expand the variable query (3) to the target $s_{\id{y}}$ will succeed,
so in this example we continue with that branch.

Because of the mapping~$\UV{x} \mapsto h$, query (3) may be relevant for resolving to the target scope~$s_{\id{y}}$.
Thus, we inspect the scope graph and search for well-formed paths to~$s_{\id{y}}$.
Since the query in question has the unification variable~$\UV{s_m}$ as its source, we need to look for paths with \emph{any} possible source.
For the graph~\cref{fig:minimod-locked-reference-example-c} in the solver state, only the one-step path from $s_{\id{A}}$ to $s_{\id{y}}$ matches the regular expression of the query.
Hence, the query is resolved by substituting $s_{\id{A}}$ for the scope variable $\UV{s_m}$.
Next, we make $s_{\id{A}}$ the new target for the hole $h$, since we assume that the source scope was not ground because the query forms part of a qualified reference; \ie, a sequence of paths.
Eventually, we solve the constraint by substituting $\UV{s_m} \mapsto s_{\id{A}}$ and $\UV{x} \mapsto \id{y}$.

\begin{figure}[H]
  \begin{HugeAngles}
    \begin{minipage}[c]{0.03\textwidth}
      \hyperref[fig:minimod-locked-reference-example-c]{$\mathcal{G}$}
    \end{minipage}
    \vline
    \begin{minipage}[c]{0.49\textwidth}
      \begin{gather*}
        \cQuery[\lblOrdLexMOD]{\svar{s}_{\id{A}}}{
          \reclos{\lblLEX}\lblMOD
        }{
          \hoPred{isMod}{\UV{a}}
        }{
          \set{\tuple{\UV{p}, \UV{a}}}
        }
        \\ 
        \hl{$s_{\id{A}}$} \cEq \mathsf{tgt}(\UV{p})
      \end{gather*}
    \end{minipage}
    \vline
    \begin{minipage}[c]{0.12\textwidth}
      \begin{gather*}
        \ \UV{a} \mapsto h
      \end{gather*}
    \end{minipage}
    \vline
    \begin{minipage}[c]{0.26\textwidth}
      \begin{gather*}
        \ h \mapsto \big( \hl{$s_{\id{A}}$} \cdot s_{\id{y}}, \UV{a}\id{.}\mkern-4mu\hl{$\mkern-4mu\id{y}\mkern-4mu$} \big)
      \end{gather*}
    \end{minipage}
  \end{HugeAngles}
\end{figure}

\noindent
In addition to refining the hole term to $\UV{a}\id{.}\id{y}$, the target scope was also refined to $s_{\id{A}}$.
The new problem to be solved is to find the qualifier that resolves to $s_{\id{A}}$.
Using the same principles as illustrated above, the remaining stuck query can be expanded.
This will solve the remaining constraints and make the hole term ground, yielding $\id{A.y}$ as the solution.

\begin{figure}[H]
  \begin{HugeAngles}
    \begin{minipage}[c]{0.03\textwidth}
      \hyperref[fig:minimod-locked-reference-example-c]{$\mathcal{G}$}
    \end{minipage}
    \vline
    \begin{minipage}[c]{0.49\textwidth}
      \centering
      $\emptyset$
    \end{minipage}
    \vline
    \begin{minipage}[c]{0.12\textwidth}
      \centering
      $\emptyset$
    \end{minipage}
    \vline
    \begin{minipage}[c]{0.26\textwidth}
      \begin{gather*}
        \ h \mapsto \big( \hl{$s_{\id{A}}$} \cdot s_{\id{A}} \cdot s_{\id{y}}, \hl{$\mkern-4mu\id{A}\mkern-4mu$}\mkern-4mu\id{.}\id{y} \big)
      \end{gather*}
    \end{minipage}
  \end{HugeAngles}
\end{figure}

\noindent
Following to our definition of a reference, the solution $\id{A.y}$ is a composite path in the scope graph from the source scope $s_{\id{A}}$,
{\color{colorblind-bright-3}$s_{\id{A}}$} \tikz[scopegraph]{\draw (0,0) edge[ref, dashdotted, color=colorblind-bright-3] (0.5,0);} {\color{colorblind-bright-3}$s_0$} \tikz[scopegraph]{\draw (0,0) edge[ref, dashdotted, color=colorblind-bright-3] (0.5,0);} {\color{colorblind-bright-3}$s_{\id{A}}$} to the scope of qualifier $\id{A}$,
followed by {\color{colorblind-bright-1}$s_{\id{A}}$} \tikz[scopegraph]{\draw (0,0) edge[ref, color=colorblind-bright-1] (0.5,0);} {\color{colorblind-bright-1}$s_{\id{y}}$} to the scope of the intended target declaration $\id{y}$, as shown here:

\begin{figure}[H]
  \vspace{-0.5\baselineskip} 
  \begin{tikzpicture}[scopegraph,node distance=2em and 2em]
    \node[scope] (s_0) {$s_0$};
    \node[scope,right= 2.5em of s_0] (s_A) {$s_{\id{A}} \mapsto \id{A}$};
    \node[scope,above right= -0.5em and 2.5em of s_A] (s_x) {$s_{\id{x}} \mapsto \id{x} : \UV{T}$};
    \node[scope,below right= -0.5em and 2.5em of s_A] (s_y) {$s_{\id{y}} \mapsto \id{y} : \tyINT$};

    \draw (s_0) edge[lbl=$\lblMOD$] (s_A);
    \draw (s_A) edge[lbl=$\lblVAR$] (s_x.west);
    \draw (s_A) edge[lbl=$\lblVAR$] (s_y.west);
    \draw (s_A) edge[lbl=$\lblLEX$, bend right] (s_0);

    \draw (s_A) edge[ref, dashdotted, bend right=50,color=colorblind-bright-3] (s_0);
    \draw (s_0) edge[ref, dashdotted, bend right,color=colorblind-bright-3] (s_A);
    \draw (s_A) edge[ref, bend right,color=colorblind-bright-1] (s_y);
  \end{tikzpicture}
  \vspace{-0.5\baselineskip} 
\end{figure}

\noindent
The \hyperref[sec:operational-semantics]{next section} formally defines the approach illustrated above.
In~\cref{sec:heuristics} we describe the heuristics we apply to make the approach usable in practice.
We evaluate our $\mathsf{synthesize}$ function on test programs with locked references in~\cref{sec:evaluation}.

\section{Operational Semantics}%
\label{sec:operational-semantics}

The previous section illustrated our approach to synthesize concrete references using an extension of the Statix solver.
This section presents an operational semantics that defines that extension.

\subsection{Syntax of Statix}%
\label{subsec:statix-syntax}

The syntax of Statix terms and constraints, defined in~\cref{fig:statix-syntax}, follows \citet{RouvoetAPKV20} and has a separation-logic-inspired flavor,
as the declarative semantics of Statix constraints is defined using separation logic.
We refer to \citeauthor{RouvoetAPKV20} for the details of this declarative semantics, and focus on the operational semantics instead in~\cref{subsec:statix-operational-semantics} and ~\cref{subsec:refsyn-operational-semantics}.

The syntax uses these distinct enumerable sets of symbols:
$\mathit{TermConstructor}$ for term constructor symbols,
$\mathit{Var}$ for term variables,
$\mathit{SetVar}$ for set variables,
$\mathit{PredSymbol}$ for names of predicates (such as $\mathsf{typeOfExpr}$ in~\cref{fig:alg-minimod-rules}),
$\mathit{Scope}$ for scope graph node identifiers,
$\mathit{Label}$ for scope graph edge labels,
$\mathit{RegEx}$ is the set of regular expressions over words comprised of label symbols, and
$\mathit{PartialOrd}$ is the set of partial orders on label symbols.

In the $\mathit{Constraint}$ syntax, $\cEmp{}$ is the trivially satisfiable constraint, akin to a $\cTrue$ constraint in a traditional logic.
Conversely, $\cFalse$ is never satisfiable.
$C_1 \cConj{} C_2$ is a \emph{separating conjunction}, where the declarative and operational semantics of Statix guarantees $C_1$ and $C_2$ construct separate scope graph fragments.
However, the reader can approximately think of $C_1 \cConj{} C_2$ as traditional logic conjunction.
The constraint $\cExists{x} C$ asserts the existence of some term named $x$, which may be referenced and constrained by $C$.
$\mathsf{single}(t, \underline{t})$ asserts that set term $\underline{t}$ is a singleton set whose element is equal to $t$, while $\cForall{x}{\underline{t}}{C}$ asserts that $C$ holds for all its elements $x$.
$\cNew t_1 \mapsto t_2$ asserts that the scope graph contains a scope identified by $t_1$ and with associated data $t_2$, whereas $t_1 \cEdge[l] t_2$ asserts that the scope identified by $t_1$ is connected via an $l$-labeled edge to the scope identified by $t_2$.

\begin{figure}[t]
  \input{fig/0400-statix-syntax}
  \caption{
    Syntax of Statix terms and constraints.
  }%
  \label{fig:statix-syntax}
\end{figure}

The syntax of queries is $\cQuery[\orderSyntax]{t}{r}{\lambda x.\: E}{z}.\: C$.
Here, the term $t$ represents a source scope term; $r$ represents a regular expression that determines reachability; $o$ is a partial order that determines visibility; $\lambda x.\: E$ is a \emph{data-well-formedness constraint} which characterizes whether a target scope and its associated data matches the query; $z$ is a set variable that will be bound in $C$ to the result of the query.
The syntax used by~\citeauthor{RouvoetAPKV20} provides a separate constraint for applying the visibility ordering $o$.
We include this ordering as a part of the query, following the syntax used by the implementation of Statix found in the Spoofax Language Workbench~\cite{KatsV10a}.\footnote{\url{https://spoofax.dev/}}

Also in contrast to~\citeauthor{RouvoetAPKV20}, we distinguish equality constraints (ranged over by $E$) from plain constraints ($C$).
This way, data well-formedness predicates of queries ($\lambda x.\: E$) use constraints that can only inspect terms and data, but cannot extend the scope graph.
Another difference from~\citeauthor{RouvoetAPKV20} is that we define the semantics of predicate constraints.
The constraint $P(t^*)$ represents an invocation of a user-specified predicate, such as those from \cref{fig:alg-minimod-rules}.
The rules we discuss next are parameterized by a specification $\mathbb{S}$ comprising rules of the form $P(t^*)$.

In the syntax of constraints and terms in \cref{fig:statix-syntax} and throughout the paper, we use $\overline{\plhdr}$ notation to represent (possibly empty) sequences, and $\plhdr^*$ notation to represent their syntax.
For example, $P(t^*)$ represents the syntax of a predicate symbol followed by a parenthesized sequence of terms.
We use $x;y$ for sequences that can be freely reordered (\eg, $x;y \approx y;x$) and $x \cdot y$ for sequences that cannot (\eg, $x \cdot y \not\approx y \cdot x$).
We overload notation and use $x;y$ and $x \cdot y$ to represent sequences both when $x$ is an element and when $x$ is a sequence, and similarly for $y$.

\subsection{Operational Semantics of Statix with Hole State Tracking}%
\label{subsec:statix-operational-semantics}

\begin{figure}[t]
  \input{fig/0400-statix-config-syntax}
  \caption{
    Syntax of Statix configurations and scope graphs.
  }%
  \label{fig:statix-config-syntax}
\end{figure}

The operational semantics in~\cref{fig:statix-operational-semantics} also follows~ \citet{RouvoetAPKV20}.
The transition relation uses the configurations whose syntax is given in \cref{fig:statix-config-syntax}.
A configuration $\opsRSState{\mathcal{G}}{\overline{C}}{U}{H}$ comprises:

\begin{itemize} 
  \item $\mathcal{G}$: the currently constructed scope graph.
  \item $\overline{C}$: the current set of constraints.
  \item $U \in (\mathit{Var} \rightharpoonup \mathit{Hole})$: associates unification variables with holes.
  This lets us determine to which hole a given constraint might relate.
  \item $H \in (\mathit{Hole} \rightharpoonup (s^* \times t))$: maps each hole in the program to its state, consisting of a list of traversed scopes $s^*$ and the term $t$ constructed so far.
\end{itemize}

\noindent
A main difference from~\citeauthor{RouvoetAPKV20} is that we extended the configuration to track the state of reference synthesis \emph{holes} via the entities $U$ and $H$.
While the rules in~\cref{fig:statix-operational-semantics} never access them, $U$ and $H$ are explicitly propagated by these rules such that substitutions resulting from unification get applied to them.
We return to the role of $U$ and $H$ in \cref{subsec:refsyn-operational-semantics}.

\begin{figure}[p]
  \centering
  \subcaptionbox*{
  }[\textwidth]{
    \input{fig/0400-statix-operational-semantics}
  }
  \subcaptionbox*{
  }[\textwidth]{
    \input{fig/0400-statix-operational-semantics-eqconstraints}
  }%
  \vspace{-1.5\baselineskip}
  \caption{Operational semantics of constraints in Statix}%
  \label{fig:statix-operational-semantics}
\end{figure}

\paragraph{Predicate Constraints}
Rule \textsc{Op-Pred} defines the semantics of \emph{predicate expansion}.
To support this rule, the transition judgment in \cref{fig:statix-operational-semantics} is parameterized by a specification $\mathbb{S}$.
This specification is given by a set of predicate rules where each rule has the shape $P(\overline{t}) \rTurnstile{} C$.
Each rule is closed (\ie, $FV(P(\overline{t}) \rTurnstile{} C) = \emptyset$), and we assume that the domain of every predicate rule is disjoint from all other rules; \ie:
\[
    \forall P\ \overline{t_1}\ \overline{t_2}\ C_1\ C_2.
    \ \left( P(\overline{t_1}) \rTurnstile{} C_1 \right) \in \mathbb{S} \wedge
     \left( P(\overline{t_2}) \rTurnstile{} C_2 \right) \in \mathbb{S} \wedge
     (\exists \theta.\: \mathsf{mgu}(\overline{t_1}, \overline{t_2}) = \theta) \Rightarrow
     \overline{t_1} \equiv \overline{t_2} \wedge C_1 \equiv C_2
\]%

\noindent
We also assume that premises that access rules in a specification $\mathbb{S}$ (\eg, $(P(\overline{t}) \rTurnstile{} C) \in \mathbb{S}$) are automatically $\alpha$-renamed to avoid variable capture.
Rule \textsc{Op-Pred} thus expands a predicate only when there exists a \emph{unique} rule $P(\overline{t_2}) \rTurnstile C$ in specification $\mathbb{S}$ whose head matches the predicate constraint $P(\overline{t_1})$.
In case the predicate constraint matches multiple rules in $\mathbb{S}$, rule \textsc{Op-Pred} does not apply.
For example, if $P$ is a predicate symbol, $f$ and $g$ are constructors, $x$, $y$ are variables, and we have a specification with rules $P(f(x)) \rTurnstile C_1$ and $P(g(x)) \rTurnstile C_2$,
then the \textsc{Op-Pred} rule does not apply to the predicate constraint $P(y)$ because $y$ can be instantiated to both $f(x)$ and $g(x)$.

The substitution yielded by \textsf{mgu} in the premise of \textsc{Op-Pred} is applied to the entire configuration after unfolding a predicate.
Here and in the rest of the paper, $\mathsf{mgu}$ is the partial function computing the most general unifier (\ie, a substitution).
We use $\bot$ to denote failure in partial functions.
We use $\theta,\theta_1,\theta_{result},\ldots$ to range over substitutions of variables by terms ($\mathit{Var} \rightharpoonup \mathit{Term}$) or set variables by set terms ($\mathit{SetVar} \rightharpoonup \mathit{SetTerm}$).
The type of substitution will be clear from the context.
The substitution functions for constraints, terms, and scope graphs are standard and elided for brevity, except for the reference entity $U$ which we describe in~\cref{subsec:refsyn-operational-semantics}.

\paragraph{Logic Constraints}
The other rules in \cref{fig:statix-operational-semantics} are directly adapted from \citet{RouvoetAPKV20}.
Rule \textsc{Op-Conj} splits a separating conjunction constraint into two constraints.
\textsc{Op-Emp} dispatches the vacuously satisfiable constraint $\cEmp{}$.
\textsc{Op-Exists} unpacks an existentially quantified constraint by choosing a fresh variable name, which may get unified using, for example, \textsc{Op-Eq}.

\paragraph{Set Constraints}
Queries yield \emph{sets} of results, so the rules in \cref{fig:statix-operational-semantics} include two dedicated constraints for matching on sets.
The semantics of $\cSingle{t_1}{\underline{t_2}}$ is given in rule \textsc{Op-Singleton} which asserts that $t_2$ must be a singleton set $\{ t' \}$, such that $t_1$ unifies with $t'$.
The semantics of constraint $\cForall{x}{\underline{t}}{C}$ is given in rule \textsc{Op-Forall}.
The rule asserts that $\underline{t}$ must be some set literal $\zeta$; \ie, a union of singleton sets.
The rule expands $\cForall{x}{\zeta}{C}$ into as many constraints as $\zeta$ has singleton sets, in each case substituting $x$ for the singleton set inhabitant.

\paragraph{Scope Graph Constraints}
The rules \textsc{Op-New-Scope} and \textsc{Op-New-Edge} create new scopes and edges in the scope graph, respectively.
Rule \textsc{Op-Data} asserts that a constraint $\cDataOf{t_1}{t_2}$ can be solved when $t_1$ is a scope $s$, and $t_2$ unifies with the term associated with scope $s$.
Rule \textsc{Op-Query} has two premises.
The function $\mathsf{Ans}$ returns the set of all paths that match the query parameters (see~\cref{subsec:scope-graphs}).
The $\mathsf{guard}$ predicate ensures that the query is \emph{guarded} in the sense that the constraints $C; \overline{C}$ do not add a new edge to the scope graph that would cause the query to yield a different answer.
Both $\mathsf{Ans}$ and $\mathsf{guard}$ are discussed in detail by \citet[\S3.1 and \S{}5.3]{RouvoetAPKV20}.

\vspace{1em}
\subsection{Operational Semantics of Reference Synthesis}%
\label{subsec:refsyn-operational-semantics}

The usual operational semantics of Statix in \cref{fig:statix-operational-semantics} is conservative about solving query and predicate constraints.
As discussed before, \textsc{Op-Pred} solves a predicate constraint only when there is exactly one possible expansion, otherwise it is \emph{stuck}.
Similarly, \textsc{Op-Query} only solves a query constraint when the source scope of the query is ground (\ie, it is a scope rather than a variable), and the $\mathsf{guard}$ premise holds.

As we can observe from the example discussed in~\cref{sec:reference-synthesis-by-example}, speculatively solving predicate and query constraints allows the Statix solver to \emph{infer} what syntax of valid references to substitute for each hole of a program.
To admit such inference, the rules in~\cref{fig:refsyn-operational-semantics} let us \emph{speculatively} expand \emph{predicate} and \emph{query constraints} in states that would otherwise be stuck.
We achieve this by introducing two new relations: $\rightarrowtail$ performs a speculative expansion step, and $\twoheadrightarrowtail$ relates sets of potentially speculatively expanded configurations.

\paragraph{The $\twoheadrightarrowtail$ Relation}
As long as the regular Statix constraint solving rules can make progress, the \textsc{Op-Solve} rule applies.
Once the solver gets stuck, \textsc{Op-Expand} applies.
The set comprehension in the bottom premise of the rule lets us speculatively solve stuck query constraints and expand predicates.
Configurations for which neither \textsc{Op-Solve} nor \textsc{Op-Expand} apply are truly stuck, and will be pruned by the set comprehension premise of \textsc{Op-Expand}.

\paragraph{Speculative Predicate Expansion}
The rule \textsc{Op-Expand-Pred} augments the plain Statix constraint solving rules from \cref{fig:statix-operational-semantics} to support selecting an \emph{arbitrary} rule for expanding a predicate constraint.

\paragraph{Speculative Query Expansion}
The rule \textsc{Op-Expand-Query} augments the plain Statix constraint solving rules from \cref{fig:statix-operational-semantics} with support for solving a stuck query constraint,
by synthesizing a path from a (possibly unknown) source scope to the current target scope of the related hole.
The rule assumes that we are resolving a reference given by a composite path, and attempts to ``prepend'' a step to the composite path.
Intuitively, if we think of composite paths as (qualified) references, this corresponds to attempting to prepend a qualifier.

\textsc{Op-Expand-Query} uses the $U$ and $H$ components of the state, to determine which hole it is expanding.
For each free variable in the program $U$ tracks to which hole it is related.
For this reason, its substitution function $U[t/x]$ has a guard that checks that each free variable in $t$ are either not related to a hole, or are related to the same hole as $x$.
As shown previously in~\cref{fig:statix-config-syntax}, the state of a hole $H(h)$ is given by a pair $(\overline{s}, t)$.
Here, $t$ is a term representing the inferred syntax for the reference, while $\overline{s}$ represents a non-empty sequence of \emph{query-connected scopes} that form its composite path.

\begin{definition}[Query-Connected Scopes]
  For a given scope graph $\mathcal{G}$, two scopes $s_1$ and $s_2$ in $\mathcal{G}$ are \emph{query-connected} by a query \smash{$q = \qBase{\orderSyntax}{s_1}{r}{\lambda x.\: E}$}, which we denote \smash{$s_1 \xxrightarrow{q} s_2$}, when there exists an $p$  such that \smash{$p \in \mathsf{Ans}(\mathcal{G}, \qBase{o}{s_1}{r}{\lambda x.\: E})$} where either $s_2 = \mathsf{tgt}(p)$ or $s_2 \in \rho_{\mathcal{G}}(\mathsf{tgt}(p))$.
  \label{def:query-connected-scopes}
\end{definition}

This uses the notation $s_2 \in \rho_{\mathcal{G}}(\mathsf{tgt}(p))$ to mean that $s_2$ is a syntactic sub-term of the data associated with $\mathsf{tgt}(p)$.
Using the definition we can define composite paths, which model references.

\begin{definition}[Composite Path]
  \label{def:composite-path}
  For a given scope graph $\mathcal{G}$, a sequence of scopes $s_0\ldots s_n$, and a sequence of queries $q_i \in \overline{q}$, a composite path in the scope graph is given by a series of query-connected scopes; \ie:
  \smash{$s_0 \xxrightarrow{q_1} s_1 \cdots \xxrightarrow{q_n} s_n$}
\end{definition}

The first premise of rule \textsc{Op-Expand-Query} asserts that the data well-formedness predicate ($\lambda y. \mkern4mu E$) has a variable occurring in it which is relevant for a hole $h$.
The head scope $s_t$ of the composite path component in the state of $h$ represents the scope that we need to connect to, in order to prepend a step to a query-connected path.
The remaining premises assert that we choose a source scope $s'$ and a target scope $s''$ that is connected to $s_t$, such the query resolves from~$s'$ to~$s''$.

\paragraph{Accepting States}
The $\mathsf{Accept}(\overline{C}, H)$ premise of \textsc{Op-Expand} holds iff (1)~$\overline{C}$ is empty (all constraints are solved), and (2) the state of each hole in $H$ has a composite path component of length $> 1$ (we have constructed a composite path for each hole).

\begin{figure}[t]
  \input{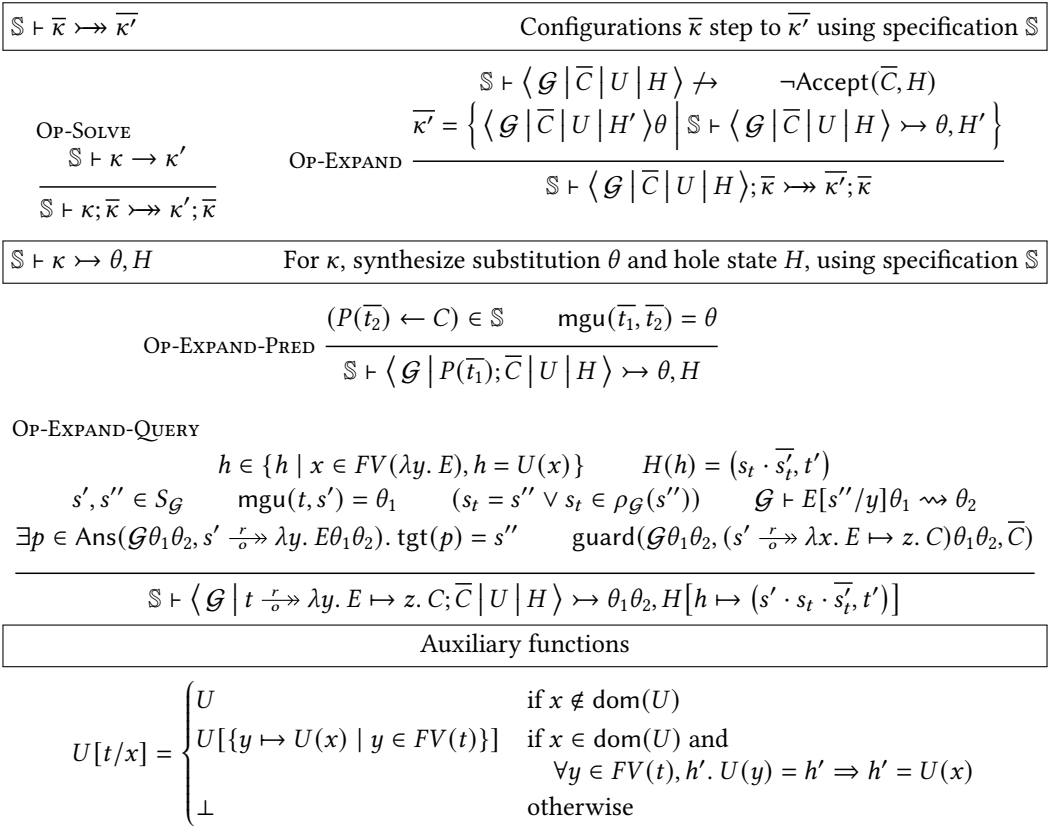}
  \captionof{figure}{
    Operational semantics of reference synthesis.
  }%
  \label{fig:refsyn-operational-semantics}
\end{figure}

\subsection{Building the Synthesize Function}%
\label{subsec:building-the-synthesizer}

Now, we have all the pieces to build the $\mathsf{synthesize}$ function:
\[
  \inferrule{
    D = \set{ x \mapsto d \mid \rhole[4]{d} \in \mathcal{P}_{\text{locked}}, x \text{ fresh} }
    \\
    U = \set{ x \mapsto h \mid x \in \mathsf{dom}(D), h \text{ fresh} }
    \\
    \theta_1 = \set{ \rhole[4]{d} \mapsto x \mid D(x) = d }
    \\
    \mathcal{P}_0 = \mathcal{P}_{\text{locked}}\theta_1
    \\
    \opsRSState
      {\emptyset}
      {P_0(\mathcal{P}_0)}
      {U}
      {\emptyset}
    \rightarrow^\ast
    \opsRSState
      {\SG_0}
      {\overline{C}}
      {U}
      {\emptyset}
    \\
    H = \set{ h \mapsto (s_d, x) \mid U(x) = h, D(x) \in \rho_{\mathcal{G}_0}(s_d) }
    \\
    \opsRSState
      {\SG_0}
      {\overline{C}}
      {U}
      {H}
    \twoheadrightarrowtail^\ast
    \overline{\kappa}
    \\
    \opsRSState
      {\SG}
      {\overline{C}'}
      {U'}
      {H'}
    \in
    \overline{\kappa}
    \\
    \mathsf{Accept} \big(\overline{C}', H' \big)
    \\
    \theta_{result} = \set{ x \mapsto t \mid H'(U(x) ) = (s_t, t) }
  }{
    \mathsf{synthesize}(\mathcal{P}_{\text{locked}}) = \mathcal{P}_0\theta_{result}
  }
\]

\noindent
We first create a fresh unification variable and hole for each locked reference in the program ($D$ and $U$, respectively),
and replace each locked reference by a fresh unification variable in the program ($\theta_1$).
Then, we solve the initial constraint ($P_0)$ on the program with holes.
From the partial scope graph in the result state, we initialize the hole state $H$ for each hole.
Then, we synthesize the references, and extract an accepted state.
From this state, we build a substitution $\theta_{result}$ that substitutes each hole variable with the synthesized reference term.

\pagebreak[4]  

\subsection{Properties}%
\label{subsec:operational-semantics-discussion}

In this section, we discuss some properties of our operational semantics: soundness, completeness, and confluence; and we discuss liveness.

\paragraph{Soundness}
First, we consider the \emph{soundness} of our approach.
Soundness consists of two components:
(1) the resulting solutions are well-typed; and
(2) each solution corresponds to a composite path in the scope graph to the target declaration.
\begin{theorem}[Soundness 1]
  \label{thm:soundness-1}
  Programs with synthesized references are well-typed.
\end{theorem}
%
\begin{theorem}[Soundness 2]
  \label{thm:soundness-2}
  Every synthesized reference corresponds to a composite path (\cref{def:composite-path}) to the target it was initially locked to.
\end{theorem}%
\iftoggle{extended}{%
  \Cref{thm:soundness-1,thm:soundness-2} have formal definitions and proofs in~\cref{subsec:well-typed-solutions,subsec:composite-paths}, respectively.
}{%
  Formal definitions and proofs of~\cref{thm:soundness-1,thm:soundness-2} can be found in the extended version of this paper~\cite{Pelsmaeker_OOPSLA25_Extended}, appendices A.2 and A.3, respectively.
}%

\paragraph{Completeness}
Ideally, we would conjecture \emph{completeness} as well.
However, completeness is hard to define, as it must rely on a generic (language-independent) definition of a reference.
For the purpose of this paper, we consider a well-formed reference in scope $s_0$ to declaration $d$ to be a composite path (\cref{def:composite-path}) through scope graph $\mathcal{G}$, from $s_0$ to $s_d$ where $d \in \rho_{\mathcal{G}}(s_d)$ (\ie, $s_d$ is associated with declaration $d$).
However, this definition is an over-approximation, as queries could be connected `by accident'.
For example, a Java expression \texttt{a.m(b)} where \id{a} is an instance of the class in which this expression occurs could be considered a reference \id{a.b} by our definition, as the target of the query for \id{a} would match the source of the query of \id{b}.
Thus, our definition is not suitable to state a completeness theorem.
In~\cref{subsec:results} we show experimental evidence that our approach is practically complete.

\paragraph{Confluence}
We also consider the property of \emph{confluence}: if two different expansions are possible, eventually, the final state sequence will be equivalent.

\begin{theorem}[Confluence]
  If\ $\overline{\kappa} \twoheadrightarrowtail \overline{\kappa}_1$~and~$\overline{\kappa} \twoheadrightarrowtail \overline{\kappa}_2$, %
  then $\exists \overline{\kappa'}.\: \overline{\kappa}_1 \twoheadrightarrowtail^\ast \overline{\kappa'} \land \overline{\kappa}_2 \twoheadrightarrowtail^\ast \overline{\kappa'} $
\end{theorem}

\begin{proof}
  This is a proof by case analysis on the expanded state:
  \begin{itemize}
    \item If different states were expanded in $\overline{\kappa}_1$ and $\overline{\kappa}_2$, the step made to obtain $\overline{\kappa}_1$ can be applied on~$\overline{\kappa}_2$ as well, and vice versa.
            This yields equivalent states again.
    \item If the same state was expanded in $\overline{\kappa}_1$ and $\overline{\kappa}_2$, either one of the following is true:
    \begin{itemize}
      \item An \textsc{Op-Solve}-step was made in both cases.
        In this case, confluence holds by virtue of~$\rightarrow$~being confluent~\citep[Theorem~4.5]{RouvoetAPKV20}.
      \item An \textsc{Op-Expand}-step was made in both cases.
        As this rule ranges over all possible expansions, $\overline{\kappa}_1 = \overline{\kappa}_2$, so confluence trivially holds.
    \end{itemize}
    The same state cannot be expanded with both \textsc{Op-Solve} and \textsc{Op-Expand}, as the first premise of \textsc{Op-Expand} requires the state to be stuck (\ie, \textsc{Op-Solve} does not apply).
  \end{itemize}
  \vspace{-1.2\baselineskip}
\end{proof}
\noindent
This confluence result is especially important for our next section, as we exploit this property to design heuristics that reduce the huge search space of possible expansions.


\paragraph{Liveness}
Finally, albeit not a property of the operational semantics, we discuss \emph{liveness} here as well.
Liveness consists of two components:
(1) the synthesis finds each solution in finite time, and
(2) when no (new) solution is available, the synthesis terminates.
Property (1) can be guaranteed by scheduling the expansion steps `fairly'; \ie in a breadth-first manner.
Property (2) requires special care for predicates that are (potentially) infinitely expanding without yielding a solution.
We return to this in~\cref{subsec:recursive-qualifiers}.

\section{Heuristics}%
\label{sec:heuristics}

\hyperref[sec:operational-semantics]{The operational semantics} is highly non-deterministic.
A direct, naive implementation would perform duplicate and unnecessary work.
To reduce work, our implementation of reference synthesis uses heuristics to guide the search.
These heuristics have three goals: (i)~to guide the search toward results, (ii)~to avoid duplicate work, and (iii)~to cut search branches early that do not lead to results.
For all these heuristics, we argue that they do not yield solutions that are not derivable from the operational semantics (soundness), and preserve all derivable solutions as well (completeness).

In this section, we discuss these heuristics, using the example in \cref{fig:example-state-for-heuristics}.
Performing reference synthesis on the program on the top-left eventually reaches the state shown in the figure.
The first two constraints are obtained from expanding rule \ref{eq:d-importok} on the initial hole on line 5 ($h_A$), represented by the variable $\UV{A_1}$.
Resolving this reference yields a scope, currently represented by the variable $\UV{s_1}$.
Once this scope is resolved, an incoming edge from $\svar{s}_{\id{B}}$ can be created.
Until then, this edge is not present in the scope graph (hence it is indicated with a dashed line).
The other two constraints correspond to the hole on line 6 ($h_x$), which is expanded to a reference with a single qualifier $\UV{A_2}\mathtt{.}\UV{x}$.
The qualifier should resolve to a scope $\UV{s_2}$, in which afterward, a query resolving to the target scope $\svar{s}_{\id{x}}$ will be performed.

\begin{figure}[t]
  \begin{minipage}[l]{0.233\textwidth}
  \begin{lstlisting}[
    language=MiniMod
  ]
    mod `\id{A}` {
      var `\id{x}` = 1
    }
    mod `\id{B}` {
      import `\rholec[5]{\id{A}}`::*
      var `\id{y}` = `\rholec[5]{\id{x}}`
    }
  \end{lstlisting}%
\end{minipage}%
\begin{minipage}{0.35\textwidth}
  \[
    H = \begin{cases}
      \UV{A_1} &\mapsto h_A
      \\
      \UV{x} &\mapsto h_x
      \\
      \UV{A_2} &\mapsto h_x
    \end{cases}
  \]
\end{minipage}%
\begin{minipage}{0.4\textwidth}
  \[
    U = \begin{cases}
      h_A &\mapsto (s_{\id{A}}, \UV{A}_1)
      \\
      h_x &\mapsto (s_{\id{x}}, \UV{A_2}\mathtt{.}\UV{x})
    \end{cases}
  \]
\end{minipage}
\\
\vspace{1em}
\begin{minipage}[c]{0.995\textwidth}
  \centering
  \begin{HugeAngles}
    \begin{minipage}[c]{0.3\textwidth}
      \centering
      \scalebox{0.9}{
        \begin{tikzpicture}[scopegraph,node distance=1em and 1em]
          \node[scope] (s_0) {$s_0$};

          \node[scope,below left=2em and 1.2em of s_0] (s_A) {$s_{\id{A}} \mapsto \id{A}$};
          \draw (s_0) edge[lbl=$\lblMOD$, bend right] (s_A);
          \draw (s_A) edge[lbl=$\lblLEX$] (s_0);

          \node[scope,below right=2em and 1.2em of s_0] (s_B) {$s_{\id{B}} \mapsto \id{B}$};
          \draw (s_0) edge[lbl=$\lblMOD$, bend left] (s_B);
          \draw (s_B) edge[lbl=$\lblLEX$] (s_0);

          \node[scope,below=3.5em of s_A.west,anchor=west] (s_x) {$s_{\id{x}} \mapsto \id{x} : \tyINT$};
          \draw (s_A) edge[lbl=$\lblVAR$, bend right] (s_x);

          \node[scope,below=3.5em of s_B.east,anchor=east] (s_y) {$s_{\id{y}} \mapsto \id{y} : \UV{T}$};
          \draw (s_B) edge[lbl=$\lblVAR$, bend left] (s_y);

          \node[scope, dashed, below = 3.5em of s_0, xshift=-0.5em] (s_var) {$\UV{s}_1$};
          \draw (s_B.west) edge[lbl=$\lblIMP$, dashed, import, out=180, in=0] (s_var.east);
        \end{tikzpicture}
      }
    \end{minipage}
    \vline
    \begin{minipage}[c]{0.52\textwidth}
      \begin{gather}
        \cUser{scopeOfMod}(\svar{s}_{\id{B}}, \UV{A_1}, \UV{s_1})
        \\
        \svar{s}_{\id{B}} \cEdge[\lblIMP] \UV{s_1}
        \\
        \cUser{scopeOfMod}(\svar{s}_{\id{B}}, \UV{A_2}, \UV{s}_2)
        \\
        \cQuery[]{\UV{s_2}}{
          \lblVAR
        }{
          \hoPred{isVar}{\UV{x}}
        }{
          \set{\tuple{\_, \UV{x} : \UV{T}}}
        }
      \end{gather}
    \end{minipage}%
    \quad  
    \vline
    \begin{minipage}{1.25em}
      \[{} \mathrel{U} {}\]
    \end{minipage}%
    \vline
    \begin{minipage}{1.25em}
      \[{} \mathrel{H} {}\]
    \end{minipage}%
  \end{HugeAngles}%
\end{minipage}
\vspace{1em}
  \caption{
    Example program (top-left) and intermediate synthesis state (below and top-right).
  }%
  \label{fig:example-state-for-heuristics}
\end{figure}
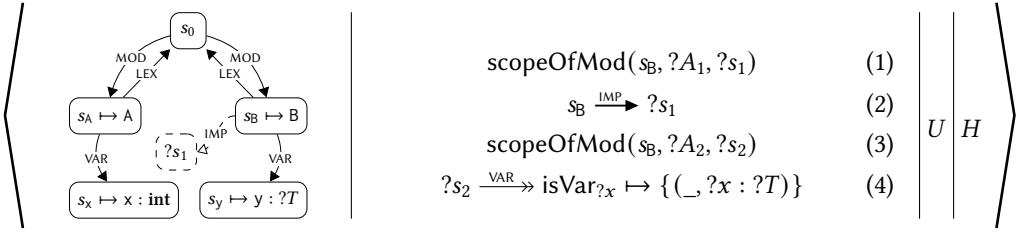%

\subsection{Selecting Constraints}%
\label{subsec:selecting-constraints}

As discussed in~\cref{subsec:operational-semantics-discussion}, our system is \emph{confluent}.
For that reason, we can choose one single order in which to expand constraints, instead of trying all possible orders.
We choose the constraint to expand according to the following criteria (considered in order):
(i)~prefer queries over predicate constraints;
(ii)~prefer predicate constraints that can lead to queries over those that cannot; and
(iii)~prefer older constraints over newer constraints.
%
The rationale behind this order is that we try to expand queries as soon as possible.
Expanding queries typically blows up the search space less than expanding predicates, and queries often reach terminal states (ground references).
Using age as a tiebreaker, we try to explore the remaining state space in a breath-first manner, to ensure we reach all possible references.
In our example, this implies that we first expand the query constraint~(4), inferring \texttt{x} to be the value of $\UV{x}$, and $s_\id{A}$ to be the target $\UV{s_2}$ of constraint~(3).
As we eventually expand all constraints, we preserve soundness and completeness.

\subsection{Expanding Queries}%
\label{subsec:expanding-queries}

When expanding a query, a source and target scope must be chosen ($s'$ and $s''$ in \textsc{Op-Expand-Query} in \cref{fig:refsyn-operational-semantics}).
Instead of trying the query with all possible combinations of source and target scopes, we can be smart about the source and target scopes we pick.
First, we choose $s''$, based on the premise that it should be equal to the current target scope, $s_t$, or contain it.
Then, we infer a unifier from the query, $\mathcal{G} \vdash E[s'' / y] \rightsquigarrow \theta$.
If no such unifier exists, we stop the search for this branch.
Finally, we traverse the scope graph backwards, starting at $s''$.
We use the inverted regular expression (\eg, $\mathsf{inv}(\reclos{\lblLEX}\, \lblVAR) = \lblVAR\, \reclos{\lblLEX}$) to guide the graph traversal to only find source scopes that have a valid path to the target scope.
This approach only over-approximates when there is another declaration, reachable from $s'$, that shadows $s''$.
It is sound, as eventually we check all the premises of the rule.
It is also complete, as all other choices for $s''$ do not satisfy the premise that relates it to $s_t$, and all other choices for $s'$ do not yield valid query answers.

In our example, when expanding the query constraint (4), the only reasonable choice for $s''$ is $s_{\id{x}}$.
Then, solving~$E$ yields $\set{\UV{x} \mapsto \id{x}}$.
Next, following the inverted regular expression ($\lblVAR$) traversing the edge $s_\id{A}~\cEdge[\lblVAR]~s_\id{x}$ backwards, we choose $s'$ to be $s_{\id{A}}$.
This instantiation is a valid query expansion.

\subsection{Expanding Predicates}%
\label{subsec:expanding-predicates}

When selecting a predicate (\cref{subsec:selecting-constraints}) we fork the computation and try each possible expansion of a predicate to a rule concurrently.
We prioritize expansions that lead to a query, as those might converge to a result quicker.
Additionally, we prioritize rules that have fewer free variables, as those add less freedom to the problem; \ie, tell us more about the final solution.
As we only prioritize certain rules over others, but never discard any, we will eventually try all rules.
\CBP{Assming fairness?}
Therefore, this does not affect completeness.
In our example, this implies that we prioritize expanding the $\mathsf{scopeOfMod}$ constraints using rule \ref{eq:s-mod} before rule \ref{eq:s-qmod}.

\subsection{Isolating Holes}%
\label{subsec:isolating-holes}

The schemes in~\cref{subsec:selecting-constraints} and~\cref{subsec:expanding-predicates} reduce the search space significantly, 
but may skew the search to holes with less complicated solutions.
Therefore, when starting the computation, we fork the state for each hole, which we call the \emph{focus hole} of that search branch.
We only expand constraints that are related to this focus hole, ensuring we make progress on this hole specifically.

However, this approach is incomplete, as sometimes the solution of one hole can only be computed after another hole is solved.
For example, for the program in~\cref{fig:example-state-for-heuristics}, \id{x} is a valid solution for $h_x$ (although in a different branch than the state presented there).
However, this can only be computed when the solution for $h_A$ is known, as it requires the edge $s_\id{B} \cEdge[\lblIMP] s_\id{A}$ to be present in the scope graph.
We need a way to ensure such `composite' solutions are computed as well.

We ensure this as follows.
Suppose query expansion (\cref{subsec:expanding-queries}) traverses an $l$-labeled edge.
When there is a scope $s$ such that future steps \emph{might} create an $l$-labeled edge in $s$ (\ie, $\overline{C} \hookrightarrow (s, l)$; see \citet[\S{}5.3]{RouvoetAPKV20}), we might miss edges.
In this case, we find the constraint that is responsible for the missing edge, and all constraints that (transitively) share a variable with this constraint.
Some of these constraints may correspond to a different hole.
Then, we peek at that hole for solutions and try insert those solutions in our current state.
Solving this state should create the missing edge.
Next, we can resume our backwards traversal.
Since we over-approximate the missing edges ($\hookrightarrow$ over-approximates, and the future edges may also have different target scopes), we preserve completeness.
In addition, we do not change the traversal itself, so we also preserve soundness.

In the example, we detect that constraint~(2) is responsible for the missing edge.
As this constraint shares a unification variable with constraint~(1), we detect it is related to $h_A$.
Thus, we insert a solution for $h_A$ (\eg, \id{A}) in the state, after which we can find the solution \id{x} for $h_x$.

\subsection{Recursive Qualifiers}%
\label{subsec:recursive-qualifiers}

Another heuristic concerns \emph{recursive qualifiers}.
We consider a reference recursive if multiple qualifiers in the reference resolve to the same declaration.
For example, consider a Java class \id{A} with field \Java|int x| and a field \Java|A a|.
In this case, a hole referring to \id{x} has an infinite number of possible solutions: \Java|x|, \Java|a.x|, \Java|a.a.x|, etc.
In these cases, each \id{a} refers to the same declaration.

To explain how we optimize synthesizing these references, we consider two of the states the synthesis of \Java|a.a.x| goes through.
After some steps, the synthesis reaches the following state:
\[
  \kappa_1 = \left\langle
    \SG
    \mkern3mu \middle| \mkern3mu \mathsf{resolveQVar}(s_{\id{A}}, \UV{q}_1, \UV{s}_1)
    \mkern3mu \middle| \mkern3mu \set{ \UV{q}_1 \mapsto h, \ldots }
    \mkern3mu \middle| \mkern3mu \set{ h \mapsto \big(s_\id{A} \cdot s_\id{x}, \UV{q}_1.\id{x}\big), \ldots }
  \right\rangle
\]

\noindent
Some forks of this branch will arrive at \Java|a.x| as a solution for this hole.
However, on the path to synthesizing \Java|a.a.x|, the following state will also be reached:
\[
  \kappa_2 = \left\langle
    \SG
  \mkern3mu \middle| \mkern3mu
    \mathsf{resolveQVar}(s_{\id{A}}, \UV{q}_2, \UV{s}_2)
  \mkern3mu \middle| \mkern3mu
    \set{ \UV{q}_2 \mapsto h, \ldots }
  \mkern3mu \middle| \mkern3mu
    \set{ h \mapsto \big(s_\id{A} \cdot s_\id{A} \cdot s_\id{x}, \UV{q}_2.\id{a}.\id{x}\big), \ldots }
  \right\rangle
\]

\noindent
These states are very similar: they share the same scope graph, have $\alpha$-equivalent constraints, and the same target scope in the hole state of $h$.
For that reason, we can conclude that \emph{each solution} for $\UV{q}_1$ is \emph{also a solution} for $\UV{q}_2$, as any derivation starting in $\kappa_1$, is also valid in $\kappa_2$.
Therefore, we can avoid synthesizing the same solution multiple times by reusing solutions for $\UV{q}_1$ when synthesizing $\UV{q}_2$.
As the traces are equivalent, we preserve soundness and completeness.

Algorithmically, we achieve this in the following three steps:
(1)~We detect pairs of solver states where
    (a)~the constraints are equivalent up to $\alpha$-renaming,
    (b)~the hole's term in the recursive solver state is an instantiation of the term in the base state, and
    (c)~the hole has the same target scope.
  For these states, we stop further synthesis on the recursive state.
(2)~For each ground solution for the variable in the base state, we emit a solution derived from the recursive state.
(3)~Repeat from step~2, using this new solution.

In the example above, this detects that $\kappa_2$ is a recursive state with respect to $\kappa_1$.
Thus, solutions from other branches rooted in $\kappa_1$ (such as \Java|a.x|) are reused in the variable $\UV{q}_2$ in $\kappa_2$, yielding \Java|a.a.x|.

This expansion is interleaved with executing other search branches, in order to guarantee liveness.
In the case no base solutions exist, this immediately terminates a search branch that would otherwise run infinitely without returning results.
This ensures that we are terminating on recursive instances that match this pattern.
When we assume the predicates that model references do not generate new scopes, only a finite number of recursive instances can be generated (because there exist a finite number of scopes in the scope graph, and a finite number of rules in a specification; hence only a finite number of states that are not equivalent according to the definition in step~1 above exist.).
That implies that we can only have non-termination when the recursive reference predicate generates fresh scopes, which typical specifications do not do.

The \hyperref[sec:evaluation]{next section} evaluates our implementation, which is based on these heuristics, providing evidence that we indeed preserve soundness and completeness.

\section{Experimental Evaluation}%
\label{sec:evaluation}

In evaluating reference synthesis, we focus on three key criteria: (1) ensuring that synthesized references are valid according to the static semantic specification of the language and resolve to the intended declaration (\emph{soundness}), (2) verify the ability to find a valid reference if it exists (\emph{completeness}), and (3) assess the efficiency of the synthesis process (\emph{performance}).
In this section we discuss how we evaluated these aspects of reference synthesis.


\pagebreak[4]
\subsection{Languages}
We evaluated our $\mathsf{synthesize}$ function on three larger programming languages: Java, ChocoPy, and Featherweight Generic Java (FGJ).

\begin{enumerate}
  \item \emph{Java} is a mainstream language with sophisticated name binding features.
  We evaluated our approach using an existing Statix specification of Java 8 from the artifact~\cite{AntwerpenV21-artifact} associated with the work of~\citet{AntwerpenV21}.
  This specification of Java unfortunately does not support generics or method references.
  We derived test cases from Java files used to validate the implementation of refactorings in the JetBrains IntelliJ IDE.\footnote{https://github.com/JetBrains/intellij-community/tree/idea/233.14808.21/java/java-tests/testData/refactoring}
  
  \item \emph{ChocoPy}~\cite{PhadyeSH19-SPLASHE} is a statically-typed dialect of Python.
  We used an existing Statix specification and ChocoPy files for our evaluation.
  
  \item \emph{Featherweight Generic Java}~\cite{IgarashiPW01} (FGJ) is a functional Java core language with full generics support.
  We used the Statix specification from the artifact~\cite{AntwerpenPRV18-artifact} of~\citet{AntwerpenPRV18}.
\end{enumerate}

\begin{figure}[t]
  \captionbox{
    Test selection.
    \label{fig:test-selection}
  }[0.55\textwidth]{
    \footnotesize
    \begin{tabular}{|l|r|r|r|}
      \hline
                            &                     Java &                     ChocoPy &                    FGJ \\
      \hline
      \hline
      Included              & \javaTotalTests{} (62\%) & \chocopyTotalTests{} (64\%) & \fgjTotalTests{} (94\%) \\
      \hline
      Negative              &                0   (0\%) &                  55  (18\%) &               0   (0\%) \\
      Java 8 incompat.      &             1076  (26\%) &                 N/A         &             N/A         \\
      Spec incompat.        &              197   (5\%) &                  10   (3\%) &               0   (0\%) \\
      No references         &              304   (7\%) &                  47  (15\%) &               3   (6\%) \\
      \hline
      Total                 & \javaOriginalTestSet{} (100\%) & \chocopyOriginalTestSet{} (100\%) & \fgjOriginalTestSet{} (100\%) \\
      \hline
    \end{tabular}%
  }%
  \hfill
  \captionbox{
    Test outcomes.
    \label{fig:test-outcomes}
  }[0.45\textwidth]{
    \footnotesize
\begin{tabular}{|l|r|r|r|}
  \hline
  &  Java  &  ChocoPy  &  FGJ  \\
  \hline
  \hline
  Success & 2382 (94\%) & 155 (79\%) & 38 (78\%) \\
  Timeout & 146 (6\%) & 41 (21\%) & 11 (22\%) \\
  Failure & 0 (0\%) & 0 (0\%) & 0 (0\%) \\
  \hline
  Total & 2528 (100\%) & 196 (100\%) & 49 (100\%) \\
  \hline
\end{tabular}
  }%
\end{figure}

\subsection{Method}
For each language we used the existing Statix semantic specifications \emph{without modification}, as our $\mathsf{synthesize}$ function works directly with these specifications.
We collected a set of test cases: single-file programs where we locked variable names, type names, and qualified member names that occur in it.
As a result, many test cases contain multiple locked references.
Each locked reference only has the target declaration as a parameter.
The original syntax of the reference is erased.
We excluded negative test cases (tests that validated that the specification or implementation gives an error on incorrect programs), those that are incompatible with our Statix specification or Java 8 (in the case of Java tests), and those that had no references to lock.
The resulting selection is shown in~\cref{fig:test-selection}.
Our reference synthesis algorithm was then applied to propose concrete references until the original reference in the program was recovered.
We set a timeout of 60~seconds per test case and ran the evaluation on a MacBook~Pro~2019 with a 2.4~GHz Intel Core~i9 processor and 64~GB memory.

\subsection{Results}%
\label{subsec:results}

A summary of the results of the evaluation is shown in~\cref{fig:test-outcomes}.
The results confirm the \emph{soundness} of our reference synthesis algorithm.
Substituting the synthesized references back into the original programs only produced well-typed programs (\cref{thm:soundness-1}).

Our evaluation demonstrates a high level of \emph{completeness}, as it successfully synthesized each reference encountered in our test cases, including non-trivial references such as Java's \Java|A.super.x| and references with ambiguous qualifiers~\cite[\S{6.5.2}]{JLS8}.
Nevertheless, formally proving \emph{completeness} of the algorithm is challenging (\cref{subsec:operational-semantics-discussion}).
%
%

\noindent
\begin{wrapfigure}{r}{7cm}
  \begin{center}
    \includegraphics[width=7cm]{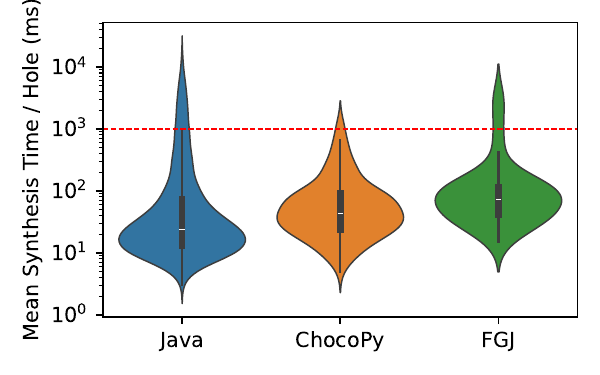}
  \end{center}%
  \caption{
    Logarithmic plot of the time spent synthesizing references per hole in each of \allSuccessfulTests{} successful test cases.
    The {\color{red}dashed line} marks 1 second.
  }%
  \label{fig:reference-synthesis-time-plot}
\end{wrapfigure}
While \emph{performance} was not the main focus of our approach, we also measured the time taken to synthesize each reference, shown as a violin plot in~\cref{fig:reference-synthesis-time-plot}.
The plot illustrates the distribution of the synthesis time per hole as a density curve, where for all three languages most results lie between 10 and 200 milliseconds.
A small box plot is depicted inside each density curve, highlighting the median, range of the central 50\% results, and range of the remaining data.
For the majority of locked references the algorithm proposed a solution within one second.
In around 7\% of the test cases, reference synthesis failed to find a solution for all locked references within the 60-second timeout.
This typically occurred when references were strongly interdependent, as reference synthesis will exhaustively explore all combinations of solutions for the dependent and the dependency (see~\cref{subsec:isolating-holes}).

\subsection{Threats to Validity}
To mitigate bias in the references we lock, or the alternatives we try, we make two conservative assumptions that are not generally true for a refactoring:
(1) \emph{all} references in the program must be locked; and
(2) the original reference syntax is unknown.
These assumptions make our test cases more challenging and may result in a worse performance than it would be in practical scenarios.

Instead of locking all references in the program, a typical refactoring would only need to lock a subset of those references, namely those that could get changed due to the program transformation.
Consequently, the likelihood of interdependent locked references would be lower.
Furthermore, most existing references are likely to remain valid despite the transformation.
Therefore, reference synthesis could prioritize verifying that the existing reference syntax still resolves to the intended declaration before synthesizing a new concrete reference.
As a result, only a handful of references need to be synthesized.


Another potential threat to validity is that our approach is parameterized by the language's static semantics written in Statix.
As shown by~\citet{RouvoetAPKV20, AntwerpenPRV18, ZwaanP24-0}, Statix can express many interesting name binding and type system concepts, but it is yet unclear what language concepts Statix cannot express.
The features of Statix that we heavily rely on are the solver interface, the presence of predicate constraints and query constraints, and conservative scheduling.
Given that we do not expect significant changes in any of these, we expect that our reference synthesis algorithm will work without fundamental modifications being required by possible future extensions to Statix.

\section{Related Work}%
\label{sec:related-work}


Refactoring as a discipline and subject goes back to the pioneering works by Griswold~\cite{Griswold:1992:PRA:144627} and Opdyke~\cite{Opdyke1992}.
Since then, refactoring has become a well-established field that has grown exponentially.
Steimann~\cite{Steimann18} notes that the growth has been so large that a recent attempt to update the Mens~and~Tourw\`{e}'s survey~\cite{MensT04} was abandoned as there were too many works to be considered.
Consequently, we focus our discussion on prior work most relevant to reference synthesis.

A distinguished line of work on implementing refactorings is Thompson et al.'s work on refactoring tools for Haskell~\cite{LiRT03,LiTR05} and Erlang~\cite{2455}, which provides support for scripting a general-purpose code transformations with possibilities to specify pre- and post-conditions that ensure the transformation preserves certain properties explicitly.
This support is realized via name-binding APIs implemented from scratch for each individual language.


As discussed in the introduction, our work builds upon previous work by~\citet{EkmanSV08,SchaferEM08,SchaferTST12,SchaferMOOPSLA2010}, who introduced the idea of tracking name-binding dependencies by ``locking'' references to declarations and then synthesizing concrete references.
Sch\"{a}fer~and~de~Moor applied this concept to synthesizing references for Java but deemed their reference synthesis algorithm to be beyond the scope of their paper~\cite[\S{2}]{SchaferMOOPSLA2010}.
Our reference synthesis algorithm is applicable to any language whose typing rules are defined using Statix~\cite{AntwerpenPRV18}.



In other closely related previous work, de~Jonge and Visser~\cite{JongeV12-LDTA} describe a language-generic API for name-binding preserving refactorings.
Their approach is inspired by the work of Ekman et~al.~\cite{EkmanSV08} on JastAdd~\cite{EkmanH07}, which they generalize by giving operations for querying name-binding information and requalifying names.
They demonstrate that this approach could be applied to multiple languages, including Stratego and a subset of Java~\cite{Visser01,VisserBT98}.
The main difference between their generic API and our work is that de~Jonge~and~Visser require the requalification function to be explicitly provided.
Our work provides this function.


Tip et al.~\cite{Tip07} demonstrate that the problem of checking that a refactoring preserves well-typedness in a language like Java can be expressed as a \emph{type constraint problem}~\cite{PalsbergSchwartzbach94}.
Their type constraint framework supports a wide range set of refactorings for a large subset of Java.
While their work focuses on Java-specific constraints, we address the broader problem of guaranteeing that references resolve to the same declaration for \emph{any} programming language whose semantics of name resolution is defined using scope graphs, with Java serving as one of the case studies.

Steimann~\cite{Steimann18} generalizes the work of Tip et al.\@ to also consider binding preservation and provides a more general and language-independent foundation for constraint-based refactoring.
Although this foundation could in principle support refactorings in terms of name-binding constraints that would guarantee binding preservation in any language, the question of how to provide a language-parametric semantics for such constraints is left open.
Our language-parametric algorithm might provide an answer to this question.

An important aspect of refactorings that move code is to avoid accidental name capture.
A common approach to avoiding name capture is \emph{renaming} (following, \eg, the Barendregt convention~\cite{0067558}).
While our reference synthesis algorithm currently focusses on synthesizing (qualified) references, it does not produce suggestions where name capture could be avoided by renaming declarations.
The language-parametric Name-Fix algorithm due to Erdweg~et~al.~\cite{ErdwegSD14} provides an interesting solution to this problem.



Despite serving a very different purpose, \citet{PelsmaekerAPV22} define a language-parametric code completion algorithm that also relies on the Statix specification.
Like our approach, they insert a free unification variable in a placeholder position, and partially type-check the program.
By inspecting the stuck constraints, they find and propose type-sound suggestions for code completion.
Their work, together with ours, shows that having a declarative but executable specification of the static semantics is essential to deriving sound, language-parametric editor services.


\pagebreak[4]  
\section{Conclusion}%
\label{sec:conclusion}

We have presented a novel approach reference synthesis that is \emph{automatic}, \emph{language-parametric}, and \emph{sound}---generating only well-typed references to the intended declarations.
This approach is applicable to any language whose static semantics is defined using typing rules in Statix~\cite{AntwerpenPRV18}.
It works out-of-the-box for such languages, is \emph{sound} by construction, and uses non-deterministic search to provide a high degree of \emph{completeness}.
Our evaluation demonstrates that our algorithm works on practical examples, but also reveals that our generic approach comes with a high performance cost, which we intend to explore in future work.

\section{Data Availability}%
\label{sec:data-availability}

This paper is accompanied by \href{https://doi.org/10.5281/zenodo.14592164}{an artifact}~\cite{Pelsmaeker_OOPSLA25_Artifact}, a Docker container that includes our reference synthesis tool, source code, and its dependencies.
Our tool can be executed in \emph{evaluation} mode to reproduce the data presented in~\cref{sec:evaluation}.
The evaluation utilizes a set of test programs with locked references, that are all included in the artifact.
Additionally, the original test sets from which the test cases were derived are also provided for further reference.

\begin{acks}
  We would like to thank Luka Miljak and the anonymous reviewers for their comments and feedback on previous versions of this paper.
  This work is supported by the Programming and Validating Software Restructurings project (\grantnum{NWOTTWMASCOT}{17933}, \grantsponsor{NWOTTWMASCOT}{NWO-TTW, MasCot}{https://www.nwo.nl/en/researchprogrammes/partnership/partnership-programmas/esi-mastering-complexity-mascot}).
\end{acks}

\bibliography{local,bibliography}


\begin{thebibliography}{42}


\ifx \showCODEN    \undefined \def \showCODEN     #1{\unskip}     \fi
\ifx \showISBNx    \undefined \def \showISBNx     #1{\unskip}     \fi
\ifx \showISBNxiii \undefined \def \showISBNxiii  #1{\unskip}     \fi
\ifx \showISSN     \undefined \def \showISSN      #1{\unskip}     \fi
\ifx \showLCCN     \undefined \def \showLCCN      #1{\unskip}     \fi
\ifx \shownote     \undefined \def \shownote      #1{#1}          \fi
\ifx \showarticletitle \undefined \def \showarticletitle #1{#1}   \fi
\ifx \showURL      \undefined \def \showURL       {\relax}        \fi
\providecommand\bibfield[2]{#2}
\providecommand\bibinfo[2]{#2}
\providecommand\natexlab[1]{#1}
\providecommand\showeprint[2][]{arXiv:#2}

\bibitem[Aho and Ullman(1972)]%
        {au72}
\bibfield{author}{\bibinfo{person}{Alfred~V. Aho} {and} \bibinfo{person}{Jeffrey~D. Ullman}.} \bibinfo{year}{1972}\natexlab{}.
\newblock \bibinfo{booktitle}{\emph{The theory of parsing, translation, and compiling}}.
\newblock \bibinfo{publisher}{Prentice-Hall, Inc.}, \bibinfo{address}{Upper Saddle River, NJ, USA}.
\newblock
\showISBNx{0-13-914556-7}


\bibitem[Barendregt(1984)]%
        {0067558}
\bibfield{author}{\bibinfo{person}{Hendrik~Pieter Barendregt}.} \bibinfo{year}{1984}\natexlab{}.
\newblock \bibinfo{booktitle}{\emph{The Lambda Calculus - Its Syntax and Semantics}}. \bibinfo{series}{Studies in Logic and the Foundations of Mathematics}, Vol.~\bibinfo{volume}{103}.
\newblock \bibinfo{publisher}{North-Holland}.
\newblock
\showISBNx{978-0-444-86748-3}


\bibitem[de~Jonge and Visser(2012)]%
        {JongeV12-LDTA}
\bibfield{author}{\bibinfo{person}{Maartje de Jonge} {and} \bibinfo{person}{Eelco Visser}.} \bibinfo{year}{2012}\natexlab{}.
\newblock \showarticletitle{A language generic solution for name binding preservation in refactorings}. In \bibinfo{booktitle}{\emph{International Workshop on Language Descriptions, Tools, and Applications, LDTA '12, Tallinn, Estonia, March 31 - April 1, 2012}}, \bibfield{editor}{\bibinfo{person}{Anthony Sloane} {and} \bibinfo{person}{Suzana Andova}} (Eds.). \bibinfo{publisher}{ACM}, \bibinfo{pages}{2}.
\newblock
\showISBNx{978-1-4503-1536-4}
\href{https://doi.org/10.1145/2427048.2427050}{doi:\nolinkurl{10.1145/2427048.2427050}}


\bibitem[Ekman and Hedin(2007)]%
        {EkmanH07}
\bibfield{author}{\bibinfo{person}{Torbjörn Ekman} {and} \bibinfo{person}{Görel Hedin}.} \bibinfo{year}{2007}\natexlab{}.
\newblock \showarticletitle{The {JastAdd} extensible {Java} compiler}. In \bibinfo{booktitle}{\emph{Proceedings of the 22nd Annual ACM SIGPLAN Conference on Object-Oriented Programming, Systems, Languages, and Applications, OOPSLA 2007, October 21-25, 2007, Montreal, Quebec, Canada}}, \bibfield{editor}{\bibinfo{person}{Richard~P. Gabriel}, \bibinfo{person}{David~F. Bacon}, \bibinfo{person}{Cristina~Videira Lopes}, {and} \bibinfo{person}{Guy L.~Steele Jr.}} (Eds.). \bibinfo{publisher}{ACM}, \bibinfo{pages}{1--18}.
\newblock
\showISBNx{978-1-59593-786-5}
\href{https://doi.org/10.1145/1297027.1297029}{doi:\nolinkurl{10.1145/1297027.1297029}}


\bibitem[Ekman et~al\mbox{.}(2008)]%
        {EkmanSV08}
\bibfield{author}{\bibinfo{person}{Torbjörn Ekman}, \bibinfo{person}{Max Schäfer}, {and} \bibinfo{person}{Mathieu Verbaere}.} \bibinfo{year}{2008}\natexlab{}.
\newblock \showarticletitle{Refactoring is not (yet) about transformation}. In \bibinfo{booktitle}{\emph{Second ACM Workshop on Refactoring Tools, WRT 2008, in conjunction with OOPSLA 2008, Nashville, TN, USA, October 19, 2008}}. \bibinfo{publisher}{ACM}, \bibinfo{pages}{5}.
\newblock
\showISBNx{978-1-60558-339-6}
\href{https://doi.org/10.1145/1636642.1636647}{doi:\nolinkurl{10.1145/1636642.1636647}}


\bibitem[Erdweg et~al\mbox{.}(2014)]%
        {ErdwegSD14}
\bibfield{author}{\bibinfo{person}{Sebastian Erdweg}, \bibinfo{person}{Tijs van~der Storm}, {and} \bibinfo{person}{Yi Dai}.} \bibinfo{year}{2014}\natexlab{}.
\newblock \showarticletitle{Capture-Avoiding and Hygienic Program Transformations}. In \bibinfo{booktitle}{\emph{ECOOP 2014 - Object-Oriented Programming - 28th European Conference, Uppsala, Sweden, July 28 - August 1, 2014. Proceedings}} \emph{(\bibinfo{series}{Lecture Notes in Computer Science}, Vol.~\bibinfo{volume}{8586})}, \bibfield{editor}{\bibinfo{person}{Richard Jones}} (Ed.). \bibinfo{publisher}{Springer}, \bibinfo{pages}{489--514}.
\newblock
\showISBNx{978-3-662-44201-2}
\href{https://doi.org/10.1007/978-3-662-44202-9_20}{doi:\nolinkurl{10.1007/978-3-662-44202-9_20}}


\bibitem[Fowler(1999)]%
        {Fowler99}
\bibfield{author}{\bibinfo{person}{Martin Fowler}.} \bibinfo{year}{1999}\natexlab{}.
\newblock \bibinfo{booktitle}{\emph{Refactoring - Improving the Design of Existing Code}}.
\newblock \bibinfo{publisher}{Addison-Wesley}.
\newblock
\showISBNx{978-0-201-48567-7}
\urldef\tempurl%
\url{http://martinfowler.com/books/refactoring.html}
\showURL{%
\tempurl}


\bibitem[Gosling et~al\mbox{.}(2015)]%
        {JLS8}
\bibfield{author}{\bibinfo{person}{James Gosling}, \bibinfo{person}{Bill Joy}, \bibinfo{person}{Guy Steele}, \bibinfo{person}{Gilad Bracha}, {and} \bibinfo{person}{Alex Buckley}.} \bibinfo{year}{2015}\natexlab{}.
\newblock \bibinfo{title}{{The {Java} {Language} {Specification} - {Java} {SE} 8 {Edition}}}.
\newblock
\urldef\tempurl%
\url{https://docs.oracle.com/javase/specs/jls/se8/html/}
\showURL{%
\tempurl}


\bibitem[Griswold(1992)]%
        {Griswold:1992:PRA:144627}
\bibfield{author}{\bibinfo{person}{William~G. Griswold}.} \bibinfo{year}{1992}\natexlab{}.
\newblock \emph{\bibinfo{title}{Program Restructuring As an Aid to Software Maintenance}}.
\newblock \bibinfo{thesistype}{Ph.\,D. Dissertation}. \bibinfo{address}{Seattle, WA, USA}.
\newblock


\bibitem[Igarashi et~al\mbox{.}(2001)]%
        {IgarashiPW01}
\bibfield{author}{\bibinfo{person}{Atsushi Igarashi}, \bibinfo{person}{Benjamin~C. Pierce}, {and} \bibinfo{person}{Philip Wadler}.} \bibinfo{year}{2001}\natexlab{}.
\newblock \showarticletitle{Featherweight Java: a minimal core calculus for Java and GJ}.
\newblock \bibinfo{journal}{\emph{ACM Transactions on Programming Languages and Systems}} \bibinfo{volume}{23}, \bibinfo{number}{3} (\bibinfo{year}{2001}), \bibinfo{pages}{396--450}.
\newblock
\href{https://doi.org/10.1145/503502.503505}{doi:\nolinkurl{10.1145/503502.503505}}


\bibitem[Kats and Visser(2010)]%
        {KatsV10a}
\bibfield{author}{\bibinfo{person}{Lennart C.~L. Kats} {and} \bibinfo{person}{Eelco Visser}.} \bibinfo{year}{2010}\natexlab{}.
\newblock \showarticletitle{The {Spoofax} language workbench}. In \bibinfo{booktitle}{\emph{Companion to the 25th Annual ACM SIGPLAN Conference on Object-Oriented Programming, Systems, Languages, and Applications, SPLASH/OOPSLA 2010, October 17-21, 2010, Reno/Tahoe, Nevada, USA}}, \bibfield{editor}{\bibinfo{person}{William~R. Cook}, \bibinfo{person}{Siobhán Clarke}, {and} \bibinfo{person}{Martin~C. Rinard}} (Eds.). \bibinfo{publisher}{ACM}, \bibinfo{pages}{237--238}.
\newblock
\showISBNx{978-1-4503-0240-1}
\href{https://doi.org/10.1145/1869542.1869592}{doi:\nolinkurl{10.1145/1869542.1869592}}


\bibitem[Li et~al\mbox{.}(2003)]%
        {LiRT03}
\bibfield{author}{\bibinfo{person}{Huiqing Li}, \bibinfo{person}{Claus Reinke}, {and} \bibinfo{person}{Simon~J. Thompson}.} \bibinfo{year}{2003}\natexlab{}.
\newblock \showarticletitle{Tool support for refactoring functional programs}. In \bibinfo{booktitle}{\emph{Proceedings of the ACM SIGPLAN Workshop on Haskell, Haskell 2003, Uppsala, Sweden, August 28, 2003}}. \bibinfo{publisher}{ACM}, \bibinfo{pages}{27--38}.
\newblock
\href{https://doi.org/10.1145/871895.871899}{doi:\nolinkurl{10.1145/871895.871899}}


\bibitem[Li et~al\mbox{.}(2006)]%
        {2455}
\bibfield{author}{\bibinfo{person}{Huiqing Li}, \bibinfo{person}{Simon Thompson}, \bibinfo{person}{László Lövei}, \bibinfo{person}{Zoltán Horváth}, \bibinfo{person}{Tamás Kozsik}, \bibinfo{person}{Anikó Víg}, {and} \bibinfo{person}{Tamás Nagy}.} \bibinfo{year}{2006}\natexlab{}.
\newblock \showarticletitle{Refactoring {Erlang} Programs}. In \bibinfo{booktitle}{\emph{The Proceedings of 12th International Erlang/OTP User Conference}}. \bibinfo{address}{Stockholm, Sweden}.
\newblock
\urldef\tempurl%
\url{http://www.cs.kent.ac.uk/pubs/2006/2455}
\showURL{%
\tempurl}


\bibitem[Li et~al\mbox{.}(2005)]%
        {LiTR05}
\bibfield{author}{\bibinfo{person}{Huiqing Li}, \bibinfo{person}{Simon Thompson}, {and} \bibinfo{person}{Claus Reinke}.} \bibinfo{year}{2005}\natexlab{}.
\newblock \showarticletitle{The {Haskell} {Refactorer}, {HaRe}, and its {API}}.
\newblock \bibinfo{journal}{\emph{Electronic Notes in Theoretical Computer Science}} \bibinfo{volume}{141}, \bibinfo{number}{4} (\bibinfo{year}{2005}), \bibinfo{pages}{29--34}.
\newblock
\href{https://doi.org/10.1016/j.entcs.2005.02.053}{doi:\nolinkurl{10.1016/j.entcs.2005.02.053}}


\bibitem[Li and Thompson(2012)]%
        {LiT12-19}
\bibfield{author}{\bibinfo{person}{Huiqing Li} {and} \bibinfo{person}{Simon~J. Thompson}.} \bibinfo{year}{2012}\natexlab{}.
\newblock \showarticletitle{Let's make refactoring tools user-extensible!}. In \bibinfo{booktitle}{\emph{Fifth Workshop on Refactoring Tools 2012, WRT '12, Rapperswil, Switzerland, June 1, 2012}}, \bibfield{editor}{\bibinfo{person}{Peter Sommerlad}} (Ed.). \bibinfo{publisher}{ACM}, \bibinfo{pages}{32--39}.
\newblock
\showISBNx{978-1-4503-1500-5}
\href{https://doi.org/10.1145/2328876.2328881}{doi:\nolinkurl{10.1145/2328876.2328881}}


\bibitem[Mens and Tourwé(2004)]%
        {MensT04}
\bibfield{author}{\bibinfo{person}{Tom Mens} {and} \bibinfo{person}{Tom Tourwé}.} \bibinfo{year}{2004}\natexlab{}.
\newblock \showarticletitle{A Survey of Software Refactoring}.
\newblock \bibinfo{journal}{\emph{IEEE Trans. Software Eng.}} \bibinfo{volume}{30}, \bibinfo{number}{2} (\bibinfo{year}{2004}), \bibinfo{pages}{126--139}.
\newblock
\urldef\tempurl%
\url{http://csdl.computer.org/comp/trans/ts/2004/02/e0126abs.htm}
\showURL{%
\tempurl}


\bibitem[Néron et~al\mbox{.}(2015)]%
        {NeronTVW15}
\bibfield{author}{\bibinfo{person}{Pierre Néron}, \bibinfo{person}{Andrew~P. Tolmach}, \bibinfo{person}{Eelco Visser}, {and} \bibinfo{person}{Guido Wachsmuth}.} \bibinfo{year}{2015}\natexlab{}.
\newblock \showarticletitle{A Theory of Name Resolution}. In \bibinfo{booktitle}{\emph{Programming Languages and Systems - 24th European Symposium on Programming, ESOP 2015, Held as Part of the European Joint Conferences on Theory and Practice of Software, ETAPS 2015, London, UK, April 11-18, 2015. Proceedings}} \emph{(\bibinfo{series}{Lecture Notes in Computer Science}, Vol.~\bibinfo{volume}{9032})}, \bibfield{editor}{\bibinfo{person}{Jan Vitek}} (Ed.). \bibinfo{publisher}{Springer}, \bibinfo{pages}{205--231}.
\newblock
\showISBNx{978-3-662-46668-1}
\href{https://doi.org/10.1007/978-3-662-46669-8_9}{doi:\nolinkurl{10.1007/978-3-662-46669-8_9}}


\bibitem[Omar et~al\mbox{.}(2017)]%
        {OmarVHAH17}
\bibfield{author}{\bibinfo{person}{Cyrus Omar}, \bibinfo{person}{Ian Voysey}, \bibinfo{person}{Michael Hilton}, \bibinfo{person}{Jonathan Aldrich}, {and} \bibinfo{person}{Matthew~A. Hammer}.} \bibinfo{year}{2017}\natexlab{}.
\newblock \showarticletitle{Hazelnut: a bidirectionally typed structure editor calculus}. In \bibinfo{booktitle}{\emph{Proceedings of the 44th ACM SIGPLAN Symposium on Principles of Programming Languages, POPL 2017, Paris, France, January 18-20, 2017}}, \bibfield{editor}{\bibinfo{person}{Giuseppe Castagna} {and} \bibinfo{person}{Andrew~D. Gordon}} (Eds.). \bibinfo{publisher}{ACM}, \bibinfo{pages}{86--99}.
\newblock
\showISBNx{978-1-4503-4660-3}
\urldef\tempurl%
\url{http://dl.acm.org/citation.cfm?id=3009900}
\showURL{%
\tempurl}


\bibitem[Opdyke(1992)]%
        {Opdyke1992}
\bibfield{author}{\bibinfo{person}{William~F. Opdyke}.} \bibinfo{year}{1992}\natexlab{}.
\newblock \emph{\bibinfo{title}{Refactoring Object-Oriented Frameworks}}.
\newblock \bibinfo{thesistype}{Ph.\,D. Dissertation}. \bibinfo{school}{University of Illinois}, \bibinfo{address}{Urbana-Champaign, IL, USA}.
\newblock Advisor(s) Ralph E. Johnson.
\newblock


\bibitem[Padhye et~al\mbox{.}(2019)]%
        {PhadyeSH19-SPLASHE}
\bibfield{author}{\bibinfo{person}{Rohan Padhye}, \bibinfo{person}{Koushik Sen}, {and} \bibinfo{person}{Paul~N. Hilfinger}.} \bibinfo{year}{2019}\natexlab{}.
\newblock \showarticletitle{ChocoPy: A Programming Language for Compilers Courses}. In \bibinfo{booktitle}{\emph{Proceedings of the 2019 ACM SIGPLAN Symposium on SPLASH-E}} \emph{(\bibinfo{series}{SPLASH-E 2019})}. \bibinfo{publisher}{Association for Computing Machinery}, \bibinfo{address}{New York, NY, USA}.
\newblock
\showISBNx{9781450369893}
\href{https://doi.org/10.1145/3358711.3361627}{doi:\nolinkurl{10.1145/3358711.3361627}}


\bibitem[Palsberg and Schwartzbach(1994)]%
        {PalsbergSchwartzbach94}
\bibfield{author}{\bibinfo{person}{Jens Palsberg} {and} \bibinfo{person}{Michael~I. Schwartzbach}.} \bibinfo{year}{1994}\natexlab{}.
\newblock \bibinfo{booktitle}{\emph{Object-oriented type systems}}.
\newblock \bibinfo{publisher}{Wiley}.
\newblock
\showISBNx{978-0-471-94128-6}


\bibitem[Pelsmaeker et~al\mbox{.}(2022)]%
        {PelsmaekerAPV22}
\bibfield{author}{\bibinfo{person}{Daniël A.~A. Pelsmaeker}, \bibinfo{person}{Hendrik van Antwerpen}, \bibinfo{person}{Casper~Bach Poulsen}, {and} \bibinfo{person}{Eelco Visser}.} \bibinfo{year}{2022}\natexlab{}.
\newblock \showarticletitle{Language-parametric static semantic code completion}.
\newblock \bibinfo{journal}{\emph{Proceedings of the ACM on Programming Languages}} \bibinfo{volume}{6}, \bibinfo{number}{OOPSLA} (\bibinfo{year}{2022}), \bibinfo{pages}{1--30}.
\newblock
\href{https://doi.org/10.1145/3527329}{doi:\nolinkurl{10.1145/3527329}}


\bibitem[Pelsmaeker et~al\mbox{.}(2025a)]%
        {Pelsmaeker_OOPSLA25}
\bibfield{author}{\bibinfo{person}{Daniel A.~A. Pelsmaeker}, \bibinfo{person}{Aron Zwaan}, \bibinfo{person}{Casper Bach}, {and} \bibinfo{person}{Arjan~J. Mooij}.} \bibinfo{year}{2025}\natexlab{a}.
\newblock \showarticletitle{Language-Parametric Reference Synthesis}.
\newblock \bibinfo{journal}{\emph{Proc. ACM Program. Lang.}} \bibinfo{volume}{9}, \bibinfo{number}{OOPSLA1}, Article \bibinfo{articleno}{123} (\bibinfo{date}{April} \bibinfo{year}{2025}).
\newblock
\href{https://doi.org/10.1145/3720481}{doi:\nolinkurl{10.1145/3720481}}


\bibitem[Pelsmaeker et~al\mbox{.}(2025b)]%
        {Pelsmaeker_OOPSLA25_Artifact}
\bibfield{author}{\bibinfo{person}{Daniel A.~A. Pelsmaeker}, \bibinfo{person}{Aron Zwaan}, \bibinfo{person}{Casper Bach}, {and} \bibinfo{person}{Arjan~J. Mooij}.} \bibinfo{year}{2025}\natexlab{b}.
\newblock \bibinfo{booktitle}{\emph{Language-Parametric Reference Synthesis (Artifact)}}.
\newblock
\href{https://doi.org/10.5281/zenodo.14592164}{doi:\nolinkurl{10.5281/zenodo.14592164}}


\bibitem[Pierce(2002)]%
        {Pierce2002}
\bibfield{author}{\bibinfo{person}{Benjamin~C. Pierce}.} \bibinfo{year}{2002}\natexlab{}.
\newblock \bibinfo{booktitle}{\emph{Types and Programming Languages}}.
\newblock \bibinfo{publisher}{MIT Press}, \bibinfo{address}{Cambridge, Massachusetts}.
\newblock


\bibitem[Poulsen et~al\mbox{.}(2018)]%
        {PoulsenRTKV18}
\bibfield{author}{\bibinfo{person}{Casper~Bach Poulsen}, \bibinfo{person}{Arjen Rouvoet}, \bibinfo{person}{Andrew~P. Tolmach}, \bibinfo{person}{Robbert Krebbers}, {and} \bibinfo{person}{Eelco Visser}.} \bibinfo{year}{2018}\natexlab{}.
\newblock \showarticletitle{Intrinsically-typed definitional interpreters for imperative languages}.
\newblock \bibinfo{journal}{\emph{Proceedings of the ACM on Programming Languages}} \bibinfo{volume}{2}, \bibinfo{number}{POPL} (\bibinfo{year}{2018}).
\newblock
\href{https://doi.org/10.1145/3158104}{doi:\nolinkurl{10.1145/3158104}}


\bibitem[Rouvoet et~al\mbox{.}(2020)]%
        {RouvoetAPKV20}
\bibfield{author}{\bibinfo{person}{Arjen Rouvoet}, \bibinfo{person}{Hendrik van Antwerpen}, \bibinfo{person}{Casper~Bach Poulsen}, \bibinfo{person}{Robbert Krebbers}, {and} \bibinfo{person}{Eelco Visser}.} \bibinfo{year}{2020}\natexlab{}.
\newblock \showarticletitle{Knowing when to ask: sound scheduling of name resolution in type checkers derived from declarative specifications}.
\newblock \bibinfo{journal}{\emph{Proceedings of the ACM on Programming Languages}} \bibinfo{volume}{4}, \bibinfo{number}{OOPSLA} (\bibinfo{year}{2020}).
\newblock
\href{https://doi.org/10.1145/3428248}{doi:\nolinkurl{10.1145/3428248}}


\bibitem[Sch{\"a}fer et~al\mbox{.}(2009)]%
        {ecoop09refactoring}
\bibfield{author}{\bibinfo{person}{Max Sch{\"a}fer}, \bibinfo{person}{Mathieu Verbaere}, \bibinfo{person}{Torbj{\"o}rn Ekman}, {and} \bibinfo{person}{Oege de Moor}.} \bibinfo{year}{2009}\natexlab{}.
\newblock \showarticletitle{Stepping Stones over the Refactoring Rubicon -- Lightweight Language Extensions to Easily Realise Refactorings}. In \bibinfo{booktitle}{\emph{23rd European Conference on Object-Oriented Programming (ECOOP '09)}}.
\newblock


\bibitem[Schäfer and de~Moor(2010)]%
        {SchaferMOOPSLA2010}
\bibfield{author}{\bibinfo{person}{Max Schäfer} {and} \bibinfo{person}{Oege de Moor}.} \bibinfo{year}{2010}\natexlab{}.
\newblock \showarticletitle{Specifying and implementing refactorings}. In \bibinfo{booktitle}{\emph{Proceedings of the 25th Annual ACM SIGPLAN Conference on Object-Oriented Programming, Systems, Languages, and Applications, OOPSLA 2010}}, \bibfield{editor}{\bibinfo{person}{William~R. Cook}, \bibinfo{person}{Siobhán Clarke}, {and} \bibinfo{person}{Martin~C. Rinard}} (Eds.). \bibinfo{publisher}{ACM}, \bibinfo{address}{Reno/Tahoe, Nevada}, \bibinfo{pages}{286--301}.
\newblock
\showISBNx{978-1-4503-0203-6}
\href{https://doi.org/10.1145/1869459.1869485}{doi:\nolinkurl{10.1145/1869459.1869485}}


\bibitem[Schäfer et~al\mbox{.}(2008)]%
        {SchaferEM08}
\bibfield{author}{\bibinfo{person}{Max Schäfer}, \bibinfo{person}{Torbjörn Ekman}, {and} \bibinfo{person}{Oege de Moor}.} \bibinfo{year}{2008}\natexlab{}.
\newblock \showarticletitle{Sound and extensible renaming for {Java}}. In \bibinfo{booktitle}{\emph{Proceedings of the 23rd Annual ACM SIGPLAN Conference on Object-Oriented Programming, Systems, Languages, and Applications, OOPSLA 2008, October 19-23, 2008, Nashville, TN, USA}}, \bibfield{editor}{\bibinfo{person}{Gail~E. Harris}} (Ed.). \bibinfo{publisher}{ACM}, \bibinfo{pages}{277--294}.
\newblock
\showISBNx{978-1-60558-215-3}
\href{https://doi.org/10.1145/1449764.1449787}{doi:\nolinkurl{10.1145/1449764.1449787}}


\bibitem[Schäfer et~al\mbox{.}(2012)]%
        {SchaferTST12}
\bibfield{author}{\bibinfo{person}{Max Schäfer}, \bibinfo{person}{Andreas Thies}, \bibinfo{person}{Friedrich Steimann}, {and} \bibinfo{person}{Frank Tip}.} \bibinfo{year}{2012}\natexlab{}.
\newblock \showarticletitle{A Comprehensive Approach to Naming and Accessibility in Refactoring Java Programs}.
\newblock \bibinfo{journal}{\emph{IEEE Trans. Software Eng.}} \bibinfo{volume}{38}, \bibinfo{number}{6} (\bibinfo{year}{2012}), \bibinfo{pages}{1233--1257}.
\newblock
\href{https://doi.org/10.1109/TSE.2012.13}{doi:\nolinkurl{10.1109/TSE.2012.13}}


\bibitem[Steimann(2018)]%
        {Steimann18}
\bibfield{author}{\bibinfo{person}{Friedrich Steimann}.} \bibinfo{year}{2018}\natexlab{}.
\newblock \showarticletitle{Constraint-Based Refactoring}.
\newblock \bibinfo{journal}{\emph{ACM Transactions on Programming Languages and Systems}} \bibinfo{volume}{40}, \bibinfo{number}{1} (\bibinfo{year}{2018}).
\newblock
\href{https://doi.org/10.1145/3156016}{doi:\nolinkurl{10.1145/3156016}}


\bibitem[Tip(2007)]%
        {Tip07}
\bibfield{author}{\bibinfo{person}{Frank Tip}.} \bibinfo{year}{2007}\natexlab{}.
\newblock \showarticletitle{Refactoring Using Type Constraints}. In \bibinfo{booktitle}{\emph{Static Analysis, 14th International Symposium, SAS 2007, Kongens Lyngby, Denmark, August 22-24, 2007, Proceedings}} \emph{(\bibinfo{series}{Lecture Notes in Computer Science}, Vol.~\bibinfo{volume}{4634})}, \bibfield{editor}{\bibinfo{person}{Hanne~Riis Nielson} {and} \bibinfo{person}{Gilberto Filé}} (Eds.). \bibinfo{publisher}{Springer}, \bibinfo{pages}{1--17}.
\newblock
\showISBNx{978-3-540-74060-5}
\href{https://doi.org/10.1007/978-3-540-74061-2_1}{doi:\nolinkurl{10.1007/978-3-540-74061-2_1}}


\bibitem[van Antwerpen et~al\mbox{.}(2016)]%
        {AntwerpenNTVW16}
\bibfield{author}{\bibinfo{person}{Hendrik van Antwerpen}, \bibinfo{person}{Pierre Néron}, \bibinfo{person}{Andrew~P. Tolmach}, \bibinfo{person}{Eelco Visser}, {and} \bibinfo{person}{Guido Wachsmuth}.} \bibinfo{year}{2016}\natexlab{}.
\newblock \showarticletitle{A constraint language for static semantic analysis based on scope graphs}. In \bibinfo{booktitle}{\emph{Proceedings of the 2016 ACM SIGPLAN Workshop on Partial Evaluation and Program Manipulation, PEPM 2016, St. Petersburg, FL, USA, January 20 - 22, 2016}}, \bibfield{editor}{\bibinfo{person}{Martin Erwig} {and} \bibinfo{person}{Tiark Rompf}} (Eds.). \bibinfo{publisher}{ACM}, \bibinfo{pages}{49--60}.
\newblock
\showISBNx{978-1-4503-4097-7}
\href{https://doi.org/10.1145/2847538.2847543}{doi:\nolinkurl{10.1145/2847538.2847543}}


\bibitem[van Antwerpen et~al\mbox{.}(2018a)]%
        {AntwerpenPRV18-artifact}
\bibfield{author}{\bibinfo{person}{Hendrik van Antwerpen}, \bibinfo{person}{Casper~Bach Poulsen}, \bibinfo{person}{Arjen Rouvoet}, {and} \bibinfo{person}{Eelco Visser}.} \bibinfo{year}{2018}\natexlab{a}.
\newblock \showarticletitle{Case Studies for Article: Scopes as Types}.
\newblock \bibinfo{journal}{\emph{Proceedings of the ACM on Programming Languages}} \bibinfo{volume}{2}, \bibinfo{number}{OOPSLA} (\bibinfo{year}{2018}).
\newblock
\href{https://doi.org/10.1145/3276915}{doi:\nolinkurl{10.1145/3276915}}


\bibitem[van Antwerpen et~al\mbox{.}(2018b)]%
        {AntwerpenPRV18}
\bibfield{author}{\bibinfo{person}{Hendrik van Antwerpen}, \bibinfo{person}{Casper~Bach Poulsen}, \bibinfo{person}{Arjen Rouvoet}, {and} \bibinfo{person}{Eelco Visser}.} \bibinfo{year}{2018}\natexlab{b}.
\newblock \showarticletitle{Scopes as types}.
\newblock \bibinfo{journal}{\emph{Proceedings of the ACM on Programming Languages}} \bibinfo{volume}{2}, \bibinfo{number}{OOPSLA} (\bibinfo{year}{2018}).
\newblock
\href{https://doi.org/10.1145/3276484}{doi:\nolinkurl{10.1145/3276484}}


\bibitem[van Antwerpen and Visser(2021a)]%
        {AntwerpenV21-artifact}
\bibfield{author}{\bibinfo{person}{Hendrik van Antwerpen} {and} \bibinfo{person}{Eelco Visser}.} \bibinfo{year}{2021}\natexlab{a}.
\newblock \showarticletitle{Scope States (Artifact)}.
\newblock \bibinfo{journal}{\emph{DARTS}} \bibinfo{volume}{7}, \bibinfo{number}{2} (\bibinfo{year}{2021}).
\newblock
\href{https://doi.org/10.4230/DARTS.7.2.1}{doi:\nolinkurl{10.4230/DARTS.7.2.1}}


\bibitem[van Antwerpen and Visser(2021b)]%
        {AntwerpenV21}
\bibfield{author}{\bibinfo{person}{Hendrik van Antwerpen} {and} \bibinfo{person}{Eelco Visser}.} \bibinfo{year}{2021}\natexlab{b}.
\newblock \showarticletitle{Scope States: Guarding Safety of Name Resolution in Parallel Type Checkers}. In \bibinfo{booktitle}{\emph{35th European Conference on Object-Oriented Programming, ECOOP 2021, July 11-17, 2021, Aarhus, Denmark (Virtual Conference)}} \emph{(\bibinfo{series}{LIPIcs}, Vol.~\bibinfo{volume}{194})}, \bibfield{editor}{\bibinfo{person}{Anders Møller} {and} \bibinfo{person}{Manu Sridharan}} (Eds.). \bibinfo{publisher}{Schloss Dagstuhl - Leibniz-Zentrum für Informatik}.
\newblock
\showISBNx{978-3-95977-190-0}
\href{https://doi.org/10.4230/LIPIcs.ECOOP.2021.1}{doi:\nolinkurl{10.4230/LIPIcs.ECOOP.2021.1}}


\bibitem[Visser(2001)]%
        {Visser01}
\bibfield{author}{\bibinfo{person}{Eelco Visser}.} \bibinfo{year}{2001}\natexlab{}.
\newblock \showarticletitle{Stratego: A Language for Program Transformation Based on Rewriting Strategies}. In \bibinfo{booktitle}{\emph{Rewriting Techniques and Applications, 12th International Conference, RTA 2001, Utrecht, The Netherlands, May 22-24, 2001, Proceedings}} \emph{(\bibinfo{series}{Lecture Notes in Computer Science}, Vol.~\bibinfo{volume}{2051})}, \bibfield{editor}{\bibinfo{person}{Aart Middeldorp}} (Ed.). \bibinfo{publisher}{Springer}, \bibinfo{pages}{357--362}.
\newblock
\showISBNx{3-540-42117-3}
\href{https://doi.org/10.1007/3-540-45127-7_27}{doi:\nolinkurl{10.1007/3-540-45127-7_27}}


\bibitem[Visser et~al\mbox{.}(1998)]%
        {VisserBT98}
\bibfield{author}{\bibinfo{person}{Eelco Visser}, \bibinfo{person}{Zine-El-Abidine Benaissa}, {and} \bibinfo{person}{Andrew~P. Tolmach}.} \bibinfo{year}{1998}\natexlab{}.
\newblock \showarticletitle{Building Program Optimizers with Rewriting Strategies}. In \bibinfo{booktitle}{\emph{Proceedings of the third ACM SIGPLAN international conference on Functional programming}}, \bibfield{editor}{\bibinfo{person}{Matthias Felleisen}, \bibinfo{person}{Paul Hudak}, {and} \bibinfo{person}{Christian Queinnec}} (Eds.). \bibinfo{publisher}{ACM}, \bibinfo{address}{Baltimore, Maryland, United States}, \bibinfo{pages}{13--26}.
\newblock
\href{https://doi.org/10.1145/289423.289425}{doi:\nolinkurl{10.1145/289423.289425}}


\bibitem[Zwaan and Poulsen(2024)]%
        {ZwaanP24-0}
\bibfield{author}{\bibinfo{person}{Aron Zwaan} {and} \bibinfo{person}{Casper~Bach Poulsen}.} \bibinfo{year}{2024}\natexlab{}.
\newblock \showarticletitle{Defining Name Accessibility Using Scope Graphs}. In \bibinfo{booktitle}{\emph{38th European Conference on Object-Oriented Programming, ECOOP 2024, September 16-20, 2024, Vienna, Austria}} \emph{(\bibinfo{series}{LIPIcs}, Vol.~\bibinfo{volume}{313})}, \bibfield{editor}{\bibinfo{person}{Jonathan Aldrich} {and} \bibinfo{person}{Guido Salvaneschi}} (Eds.). \bibinfo{publisher}{Schloss Dagstuhl - Leibniz-Zentrum für Informatik}.
\newblock
\showISBNx{978-3-95977-341-6}
\href{https://doi.org/10.4230/LIPIcs.ECOOP.2024.47}{doi:\nolinkurl{10.4230/LIPIcs.ECOOP.2024.47}}


\bibitem[Zwaan and van Antwerpen(2023)]%
        {ZwaanA23}
\bibfield{author}{\bibinfo{person}{Aron Zwaan} {and} \bibinfo{person}{Hendrik van Antwerpen}.} \bibinfo{year}{2023}\natexlab{}.
\newblock \showarticletitle{Scope Graphs: The Story so Far}. In \bibinfo{booktitle}{\emph{Eelco Visser Commemorative Symposium, EVCS 2023, April 5, 2023, Delft, The Netherlands}} \emph{(\bibinfo{series}{OASIcs}, Vol.~\bibinfo{volume}{109})}, \bibfield{editor}{\bibinfo{person}{Ralf Lämmel}, \bibinfo{person}{Peter~D. Mosses}, {and} \bibinfo{person}{Friedrich Steimann}} (Eds.). \bibinfo{publisher}{Schloss Dagstuhl - Leibniz-Zentrum für Informatik}.
\newblock
\showISBNx{978-3-95977-267-9}
\href{https://doi.org/10.4230/OASIcs.EVCS.2023.32}{doi:\nolinkurl{10.4230/OASIcs.EVCS.2023.32}}


\end{thebibliography}



\begin{thebibliography}{3}


\ifx \showCODEN    \undefined \def \showCODEN     #1{\unskip}     \fi
\ifx \showDOI      \undefined \def \showDOI       #1{#1}\fi
\ifx \showISBNx    \undefined \def \showISBNx     #1{\unskip}     \fi
\ifx \showISBNxiii \undefined \def \showISBNxiii  #1{\unskip}     \fi
\ifx \showISSN     \undefined \def \showISSN      #1{\unskip}     \fi
\ifx \showLCCN     \undefined \def \showLCCN      #1{\unskip}     \fi
\ifx \shownote     \undefined \def \shownote      #1{#1}          \fi
\ifx \showarticletitle \undefined \def \showarticletitle #1{#1}   \fi
\ifx \showURL      \undefined \def \showURL       {\relax}        \fi
\providecommand\bibfield[2]{#2}
\providecommand\bibinfo[2]{#2}
\providecommand\natexlab[1]{#1}
\providecommand\showeprint[2][]{arXiv:#2}

\bibitem[Gosling et~al\mbox{.}(2015)]%
        {JLS8}
\bibfield{author}{\bibinfo{person}{James Gosling}, \bibinfo{person}{Bill Joy}, \bibinfo{person}{Guy Steele}, \bibinfo{person}{Gilad Bracha}, {and} \bibinfo{person}{Alex Buckley}.} \bibinfo{year}{2015}\natexlab{}.
\newblock \bibinfo{title}{{The {Java} {Language} {Specification} - {Java} {SE} 8 {Edition}}}.
\newblock
\newblock
\urldef\tempurl%
\url{https://docs.oracle.com/javase/specs/jls/se8/html/}
\showURL{%
\tempurl}


\bibitem[Igarashi et~al\mbox{.}(2001)]%
        {IgarashiPW01}
\bibfield{author}{\bibinfo{person}{Atsushi Igarashi}, \bibinfo{person}{Benjamin~C. Pierce}, {and} \bibinfo{person}{Philip Wadler}.} \bibinfo{year}{2001}\natexlab{}.
\newblock \showarticletitle{Featherweight Java: a minimal core calculus for Java and GJ}.
\newblock \bibinfo{journal}{\emph{ACM Transactions on Programming Languages and Systems}} \bibinfo{volume}{23}, \bibinfo{number}{3} (\bibinfo{year}{2001}), \bibinfo{pages}{396--450}.
\newblock
\urldef\tempurl%
\url{https://doi.org/10.1145/503502.503505}
\showDOI{\tempurl}


\bibitem[Padhye et~al\mbox{.}(2019)]%
        {PhadyeSH19-SPLASHE}
\bibfield{author}{\bibinfo{person}{Rohan Padhye}, \bibinfo{person}{Koushik Sen}, {and} \bibinfo{person}{Paul~N. Hilfinger}.} \bibinfo{year}{2019}\natexlab{}.
\newblock \showarticletitle{ChocoPy: A Programming Language for Compilers Courses}. In \bibinfo{booktitle}{\emph{Proceedings of the 2019 ACM SIGPLAN Symposium on SPLASH-E}} \emph{(\bibinfo{series}{SPLASH-E 2019})}. \bibinfo{publisher}{Association for Computing Machinery}, \bibinfo{address}{New York, NY, USA}.
\newblock
\showISBNx{9781450369893}
\urldef\tempurl%
\url{https://doi.org/10.1145/3358711.3361627}
\showDOI{\tempurl}


\end{thebibliography}

\iftoggle{extended}{

\pagebreak[4]  

\appendix

\section{Soundness Proofs}
After some preliminaries (\cref{subsec:preliminaries}), in this appendix we prove~\cref{thm:soundness-1}: that the reference synthesis algorithm only computes well-typed solutions (\cref{subsec:well-typed-solutions}).
Next, we prove~\cref{thm:soundness-2}: that the solutions it computes correspond to a composite path in the scope graph (\cref{subsec:composite-paths}).

\subsection{Preliminaries}%
\label{subsec:preliminaries}

\paragraph{Notation}
To simplify dealing with variable renaming and capture, for these proofs we assume Barendregt's convention.
Where relevant, we present the constraint set in a solver state in the order we intend to solve them.
Moreover, we are overloading some notation for brevity:
\begin{itemize}
  \item We write $\kappa \rightarrowtail \kappa'$ to denote $\exists \theta\, H'.\ \kappa \rightarrowtail \theta, H' \land \kappa' = \opsRSState{\SG}{\overline{C}}{U}{H'}\theta$\\(\ie, $\kappa$ can be expanded into $\kappa'$, ignoring other expansions).
  \item Similarly, we write $\kappa \twoheadrightarrowtail \kappa'$ to denote $\exists \overline{\kappa}'.\ \kappa \twoheadrightarrowtail \overline{\kappa}' \land \kappa' \in \overline{\kappa}'$\\(\ie, $\kappa$ can be step into $\kappa'$, ignoring other steps).
\end{itemize}
We use the $\oplus$ operator to manipulate the components of a solver state:
\[
  \kappa \oplus C = \opsRSState{\SG}{\overline{C}; C}{U}{H}
\]
Finally, we use $\SG_\kappa$, $\overline{C}_\kappa$, $U_\kappa$ and $H_\kappa$ to denote the components of a solver state $\kappa$,
or $\SG_i$, $\overline{C}_i$, $U_i$ and~$H_i$ to denote the components of a solver state $\kappa_i$, and similar for primes.

In these proofs, we will often use exhaustive application of the step relations, written with a bullet symbol ($\bullet$) and defined as follows:

\begin{definition}[Exhaustive Application]
  \[
    \inferrule{
      \kappa_1 \rightarrow^{\ast} \kappa_2 \\ \kappa \not\rightarrow
    }{
      \kappa_1 \rightarrow^{\bullet} \kappa_2
    }
  \]
  and likewise for $\rightarrowtail^{\bullet}$ and $\twoheadrightarrowtail^{\bullet}$.
\end{definition}

The $\kappa \not\rightarrow$ premise states that there are no further steps possible from $\kappa$, ensuring exhaustiveness.
The $\rightarrow^{\bullet}$ relation is only differs from the reflexive transitive closure regarding this additional premise.
As such, we will freely convert between these relations when applicable.

Next, we define which initial constraints we consider well-formed:

\begin{definition}[Well-formed Initial Constraint]
  \label{conj:wellformed-spec}
  \begin{mathpar}
    \inferrule{
      \forall \mathcal{P}\, \kappa.\
        \left(
        \mathit{FV}(\mathcal{P}) = \emptyset \land
          \opsRSState
            {\emptyset}
            {P_0(\mathcal{P})}
            {\emptyset}
            {\emptyset}
          \rightarrow^{\bullet}
            \kappa
        \right)
        \implies
          \overline{C}_\kappa = \emptyset \lor \overline{C}_\kappa = \set{\cFalse}
    }{
      \mathbb{S} \vdash P_0 \mathrel{\mathbf{init}}
    }
  \end{mathpar}
\end{definition}

This specification states that, for any \emph{ground} input program $\mathcal{P}$, the solver will either give an accepting configuration ($C_\kappa = \emptyset$),
or a rejecting configuration ($C_\kappa = \set{\cFalse}$).
That is, we exclude stuckness from being a possible result of solving the initial constraint \emph{on a ground program}.

\pagebreak[4]  

Second, along the lines of definition 4.1 by~\citet{RouvoetAPKV20}, we define the embedding of a substitution $\theta$ as the embedding of the conjunction of the equalities that it represents:

\begin{definition}[Embedding of Subtitution]
\[
  \llbracket \theta \rrbracket = \left( \mathop{\stBigAnd}\limits_{x \in \mathsf{dom}(\theta)} x \cEq x \theta  \right)
\]
\end{definition}

We use this embedding to use constraints and substitution interchangeably:

\begin{lemma}[Equivalence of Embedding and Applying Substitutions]
  \label{lem:embedding-equality}
  \begin{mathpar}
    \kappa_1 \oplus \llbracket \theta \rrbracket
    \rightarrow^{\ast}
    \kappa_2
    \land
    \overline{C}_1 \theta = \overline{C}_2
    \iff
    \kappa_1
    \theta \approx \kappa_2
  \end{mathpar}
\end{lemma}

\begin{proof}
  By induction on $\rightarrow$, using rules \textsc{Op-Conj} and \textsc{Op-Eq-True}, using the fact that unification with a free variable (such as generated by the embedding) on one side cannot fail.
\end{proof}

Finally, we state lemmas about weakening and strengthening solver trace.
First, we state that removing equality constraints from a successful solver trace cannot yield failing solver traces (it either gets stuck, or it is successful in which case the equality constraints were redundant):

\begin{lemma}[Equality Weakening]
  \label{lem:weakening}
  \begin{mathpar}
    \forall t_1\, t_2\, \kappa_1\, \kappa_2\, \kappa_2'.\ %
      \left(
      \kappa_1 \oplus (t_1 \cEq t_2)
      \rightarrow^{\bullet}
      \kappa_2
      \land
      \overline{C}_2 = \emptyset
      \right)
    \implies
      \left(
      \kappa_1
      \rightarrow^{\bullet}
      \kappa_2'
      \land
        \overline{C}_2' \neq \set{\cFalse}
      \right)
  \end{mathpar}
\end{lemma}

\begin{proof}
  The presence of equality constraints cannot block other constraints from being solved.\footnote{\label{fn:blocking}%
    This can be proved by the fact that the only negative condition on the constraint set occurs in the $\mathsf{guard}$ premise of the \textsc{Op-Query} rule.
    However, equality constraints do not affect $\mathsf{guard}$, hence all constraints in $\overline{C}$ can still be solved.
    On the other hand, the $\mathsf{guard}$ premise prevents the general version of this lemma to be true.
    This is expected, as \eg strengthening a constraint set with edge assertions could invalidate earlier queries.
  }
  Thus, if a failure state can be reached from the initial state on the right-hand side, it must have been reachable from the initial state on the left-hand side as well, by not selecting $t_1 \cEq t_2$ to be solved.
  Then, by confluence, any trace for the initial state should reach the failure state.
  This contradicts the assumption that the first trace was successful.
\end{proof}

Second, we use the lemma that strengthening a solver trace with redundant equalities (\ie, equalities that hold in the final state anyway) preserves satisfiability:

\begin{lemma}[Strengthening with Redundant Equalities]
    \label{lem:strengthening-eq}
    \begin{mathpar}
        \forall \kappa_1\, \kappa_2\, \theta.\ %
            \kappa_1
            \rightarrow^{\bullet}
            \kappa_2
        \land
            \kappa_2 \theta = \kappa_2
        \implies
            \kappa_1 \oplus \llbracket \theta \rrbracket
            \rightarrow^{\ast}
                \kappa_2
    \end{mathpar}
\end{lemma}
\begin{remark}
    This lemma assumes that the operational semantics tracks which variables we have unified and substituted already during constraint solving.
    Without this assumption, $\theta$ may contain information that contradicts $\kappa_1$, but (due to substitution) not represented in $\kappa_2$.
    For example, when we solve a constraint \smash{$f(x) \cEq f(g())$}, for some variable $x$ and term constructors $f$ and $g$, we substitute $x$ for~$g()$, and remember in all future states that $x$ was substituted for $g()$.
    The $\theta$ in the premise of the lemma below is assumed to be \emph{consistent} with the variables that we have substituted for already, where by consistent we mean the following.
    $\theta'$ is consistent with $\theta$ when $\forall x \in \mathsf{dom}(\theta).\ x\in \mathsf{dom}(\theta') \implies \mathsf{mgu}(\theta(x), \theta'(x))\neq \bot$.
    The operational semantics does not currently track information about which variables have been unified explicitly, but can be adapted to do so.
\end{remark}

\pagebreak[4]  

\begin{proof}
    The new equalities do not prevent $\overline{C}_1$ from being solved.\footnotemark[\getrefnumber{fn:blocking}]
    Thus, the following trace exists:
    \begin{mathpar}
      \kappa_1 \oplus \llbracket \theta \rrbracket
      \rightarrow^{\ast}
      \kappa_2 \oplus \llbracket \theta \rrbracket
    \end{mathpar}
    Then, because equalities can be solved unconditionally, we can solve those next.
    \begin{mathpar}
      \kappa_2 \oplus \llbracket \theta \rrbracket
      \rightarrow^{\ast}
      \kappa_2' = \opsRSState
        {\SG_2'}
        {\overline{C}_2 \theta}
        {U_2'}
        {H_2'}
    \end{mathpar}
    Then, by~\cref{lem:embedding-equality}, $\kappa_2' = \kappa_2 \theta$.
    By the assumption that $\kappa_2 \theta = \kappa_2$, we have that $\kappa_2' = \kappa_2$.
    Therefore, a trace $\kappa_1 \oplus \llbracket \theta \rrbracket \rightarrow^{\ast} \kappa_2$ exists.
\end{proof}

\subsection{Generated Proposals are Well-Typed}%
\label{subsec:well-typed-solutions}

Now, we prove our first soundness property: the algorithm only computes well-typed solutions.
\begin{theorem*}[Soundness 1: Well-Typed Solutions]
  \label{thm:soundness-1-appendix}
  \begin{alignat*}{1}
    \forall\mathbb{S}\, P_0\, \mathcal{P}\, \mathcal{P}'.\ & \mathbb{S} \vdash P_0 \mathrel{\mathbf{init}}\\
    \implies & \mathsf{synthesize}(\mathcal{P}) = \mathcal{P}'\\
    \implies & \exists\SG.\ %
      \mathbb{S} \vdash \opsRSState
        {\emptyset}
        {P_0(\mathcal{P}')}
        {\emptyset}
        {\emptyset}
      \rightarrow^{\bullet}
      \opsRSState
        {\SG}
        {\emptyset}
        {\emptyset}
        {\emptyset}
  \end{alignat*}
\end{theorem*}

This property is stating that applying the initial constraint for any synthesized program $\mathcal{P}'$ is satisfiable;
\ie, can be reduced to a \emph{successful} solver state.

\begin{proof}
    We will prove this theorem by transforming the synthesis trace for $\mathcal{P}$ in a valid solver trace for $\mathcal{P}'$.
    Observe that each valid synthesis can be considered a trace with the following shape:
    \[
        \kappa_1 = \opsRSState
        {\emptyset}
        {P_0(\mathcal{P})}
        {U_1}
        {H_1}
        \rightarrow^{\bullet}
        \kappa_2
        \rightarrowtail
        \kappa_2'
        \twoheadrightarrowtail^{\bullet}
        \kappa_n = \opsRSState
        {\SG_n}
        {\emptyset}
        {U_n}
        {H_n}
    \]
    where $\kappa_2' = \kappa_2\theta_2$, given that $\theta_2$ is the substitution that was computed in the expand step $\kappa_2 \rightarrowtail \kappa_2'$.
    By~\cref{lem:embedding-equality}, we can create a new trace as follows:
    \[
        \kappa_2 \oplus {\color{red} \llbracket \theta_2 \rrbracket}
        \rightarrow^{\bullet}
        \kappa_2'
        \twoheadrightarrowtail^{\bullet}
        \kappa_n
    \]
    By repeating this on all expand steps, we can create the following sequence of traces:
    \begin{alignat*}{1}
        \kappa_1 &\rightarrow^{\bullet} \kappa_2 \\
        \kappa_2 \oplus \llbracket \theta_2 \rrbracket &\rightarrow^{\bullet} \kappa_3 \\
        &\ldots \\
        \kappa_{n-1} \oplus \llbracket \theta_{n-1} \rrbracket &\rightarrow^{\bullet} \kappa_n
    \end{alignat*}
    These traces are all valid solver traces (i.e., each step is valid according to one of the \textsc{Op-*} rules).

    Next, we use $\theta_{\text{result}}$ from $\mathsf{synthesize}$ to eliminate all intermediate equalities $\theta_2 \ldots \theta_{n-1}$.
    Because $\theta_{\text{result}}$ is derived from $U_n$ and $H_n$, $\theta_{\text{result}}$ is redundant in $\kappa_n$ (\ie, $\kappa_n \theta_{\text{result}} = \kappa_n$).
    When prepending the embedding of this unifier to the last solver trace, by~\cref{lem:strengthening-eq}, we construct the following trace:%
    \footnote{
        Note that the variables in $\theta_{\text{result}}$ are free in the original constraint.
        For that reason, assuming the state tracks non-free variables (see the remark at~\cref{lem:strengthening-eq}), we can freely move around this constraint.
    }
    \[
        \kappa_{n-1} \oplus \llbracket \theta_{n-1} \rrbracket \oplus
            {{\color{red} \llbracket \theta_{\text{result}} \rrbracket}}
        \rightarrow^{\bullet}
        \kappa_n
    \]
    Using~\cref{lem:weakening} (removing $\llbracket \theta_{n-1} \rrbracket$), Statix' confluence, and the fact that all variables in $\theta_{\text{result}}$ are free in the original constraint, we infer that $\kappa_{n-1} \oplus
    {{\llbracket \theta_{\text{result}} \rrbracket}}$ either steps to $\kappa_4$, or gets stuck.
    By applying the same reasoning to the previous traces, we can conclude that $\kappa_1 \oplus {\llbracket \theta_{\text{result}} \rrbracket}$ either steps to $\kappa_4$, or gets stuck.

    Using confluence and~\cref{lem:embedding-equality}, together with the fact that ${\llbracket \theta_{\text{result}} \rrbracket}$ can be solved first, we infer the following trace exists:
    \[
        \kappa_1 \oplus {\llbracket \theta_{\text{result}} \rrbracket}
        \rightarrow^{\ast}
        \kappa_1' = \opsRSState
        {\emptyset}
        {\overline{C}_1\theta_{\text{result}}}
        {U_1'}
        {H_1'}
        \rightarrow^{\bullet}
        \kappa_n'
    \]
    where $\kappa_n' = \kappa_n$ or some stuck state.
    Note that $\overline{C}_1\theta_{\text{result}} = {P_0(\mathcal{P}')}$,
    and hence the intermediate state~$\kappa_1'$ is exactly the initial state for solving $P_0(\mathcal{P}')$.

    Finally, by our assumption that the initial constraint cannot get stuck on a ground program, $\kappa_n'$ cannot be stuck.
    Hence $\kappa_n' = \kappa_n$, which is a successful state.
    For that reason, we conclude that $\mathcal{P}'$ is well-typed.

\end{proof}



\subsection{Synthesized References correspond to Composite Paths}%
\label{subsec:composite-paths}

Next, we prove our second soundness property: the synthesized references correspond to a composite path to the target.
First, we restate our definition of composite path as a judgment $\kappa \vdash s^{\ast} \ \textsc{cp}$, stating that $s^{\ast}$ is a composite path in $\kappa$.
\begin{definition}[Composite Path]
  \[
    \inferrule[\textsc{CP-Single}]{
    }{
      \kappa \vdash s \ \textsc{cp}
    }
    \qquad
    \inferrule[\textsc{CP-Comp}]{
      q = \qBase{o}{s_1}{r}{\lambda x.\, E} \textit{ is solved in } \kappa
      \\
      p \in \mathsf{Ans}(\SG_\kappa, q)
      \\
      \left(\mathsf{tgt}(p) = s_2 \lor s_2 \in \rho_{\SG_{\kappa}}(\mathsf{tgt}(p)) \right)
      \\
      \kappa \vdash s_2 \cdot s^{*} \ \textsc{cp}
    }{
      \kappa
      \vdash
        s_1 \cdot s_2 \cdot s^{*} \ \textsc{cp}
    }
  \]
\end{definition}
\begin{remark}
    The phrase `$q$ is solved in $\kappa$' can be made precise by making the state transitions ($\rightarrow$) track solved constraints explicitly.
\end{remark}
We use this definition to define well-formed hole states $H$:
\begin{definition}[Well-formed $H$]
  \[
    \inferrule[\textsc{H-WF}]{
      \forall (s^{\ast}, t) \in \mathsf{ran}(H_\kappa).\ \kappa \vdash s^{\ast} \  \textsc{cp}
    }{
      \kappa \  \mathsf{wf}_H
    }
  \]
\end{definition}

This definition states that a hole state $H_\kappa$ is well-formed if for each hole state $(s^{\ast}, t) \in \mathsf{ran}(H_\kappa)$, $s^{\ast}$ is a composite path in $\kappa$.
With this definition of composite paths, we can state our second soundness property as follows.

\begin{theorem}[Soundness 2: Solutions correspond to Composite Paths]
  \label{thm:soundness-2-appendix}
  \begin{mathpar}
  \forall \kappa_0\, \kappa_1.\ %
      \left( \kappa_0 \textit{ initialized by } \mathsf{synthesize} \land \kappa_0 \twoheadrightarrowtail^{\bullet}  \kappa_1 \right)
    \implies
      \kappa_1 \ \mathsf{wf}_H
  \end{mathpar}
\end{theorem}
This theorem states that every synthesis result yields hole states $(s^{\ast}, t) \in \mathsf{ran}(H_1)$ such that $s^{\ast}$ forms a composite path.
To prove this theorem, we need the lemma that $\rightarrow$ preserves hole state well-formedness.
\begin{lemma}
  \label{lem:solve-preserves-wfh}
  \begin{mathpar}
    \forall \kappa_0\, \kappa_1.\ %
      \kappa_0 \ \mathsf{wf}_H
      \land
      \kappa_0 \rightarrow \kappa_1
    \implies
      \kappa_1 \ \mathsf{wf}_{H}
  \end{mathpar}
\end{lemma}
\begin{proof}
    This can be proven by case analysis on $\rightarrow$, where none of the rules modify the hole state (apart from substitution, which does not affect the $s^{\ast}$ component, as it is ground).
\end{proof}

\pagebreak[4]  

Ideally, we would prove a similar lemma for $\twoheadrightarrowtail$ as well.
However, this cannot be done, as \textsc{Op-Expand-Query} prepends a scope to the composite path in $H(h)$, for which the corresponding query is not yet solved.
Instead, we prove the following lemma, that states that, after an expand-step, subsequent solving ensures that $H$ will \emph{eventually} be well-formed.

\begin{lemma}
  \label{lem:expand-eventually-preserves-wfh}
  \begin{mathpar}
    \forall \kappa_0\, \kappa_1\, \kappa_2.\ %
      \kappa_0 \ \mathsf{wf}_{H}
      \land
      \kappa_0 \rightarrowtail \kappa_1 \rightarrow^{\bullet} \kappa_2
    \implies
      \kappa_2 \ \mathsf{wf}_{H}
  \end{mathpar}
\end{lemma}

\begin{proof}
  This lemma can be proven by case analysis on the $\kappa_0 \rightarrowtail \kappa_1$ step:
  \begin{itemize}
    \item \textsc{Op-Expand-Pred}: This rule does not affect $H$, hence $\kappa_1 \ \mathsf{wf}_{H}$ and by~\cref{lem:solve-preserves-wfh}, $\kappa_2 \ \mathsf{wf}_{H}$.
    \item \textsc{Op-Expand-Query}: This rule expands a query \smash{$q = \qBase{o}{s'}{r}{\lambda x.\, E}$}.
      As a result, it prepends $s'$ to the composite path \smash{$(s'' \cdot s^{\ast}, t) \in \mathsf{ran}(H_0)$}.
      When proving $\kappa_1 \ \mathsf{wf}_{H}$, by inversion using \textsc{H-WF} and \textsc{CP-Comp}, we have the following proof obligations:
      \begin{itemize}
          \item $q$ is solved in $\kappa_1$,
          \item $p \in \mathsf{Ans}(\SG_1, q) \land (\mathsf{tgt}(p) = s'' \lor s'' \in \rho_{\SG_1}(\mathsf{tgt}(p)))$, and
          \item $\kappa_1 \vdash s'' \cdot s^{*} \ \textsc{cp}$.
      \end{itemize}
      The second obligation is a direct consequences of the \textsc{Op-Expand-Query} rule, and the last obligation is a direct consequence of the $\kappa_0 \ \mathsf{wf}_{H}$ assumption.
      However, the first premise is not yet satisfied, as \textsc{Op-Expand-Query} does not actually solve $q$.

      To prove that $\kappa_2 \ \mathsf{wf}_{H}$, we need to show that $q$ is solved in $\kappa_2$.
      To this end, observe that the premises of \textsc{Op-Expand-Query} entail the premises of \textsc{Op-Query} for $q$.
      Hence, $q$ is not \emph{blocked}; it \emph{can} be solved in $\kappa_1$.
      Thus, a state $\kappa_1'$ exists such that $\kappa_1 \rightarrow \kappa_1'$ and $q$ is solved in $\kappa_1'$.
      For that reason, $\kappa_1' \ \mathsf{wf}_{H}$.
      Now, consider the state $\kappa_2'$ obtained after solving $\kappa_1'$ completely (\ie, $\kappa_1' \rightarrow^{\bullet} \kappa_2'$).
      By~\cref{lem:solve-preserves-wfh}, $\kappa_2' \ \mathsf{wf}_{H}$.
      Finally, by confluence of the operational semantics, $\kappa_2' \approx \kappa_2$, and hence $\kappa_2 \ \mathsf{wf}_{H}$.
  \end{itemize}
\end{proof}

With these lemmas in place, we can prove~\cref{thm:soundness-2}:
\begin{proof}
  Observe that any synthesis trace $\kappa_0 \twoheadrightarrowtail^{\bullet} \kappa_n$ can be decomposed into parts that start with an expand step, and a sequence of solve steps:
  \[
    \left[\kappa_{i-1} \rightarrowtail \kappa_i' \rightarrow^{\bullet} \kappa_i  \right]_{i \in 1\ldots n}
  \]
  We can prove the theorem using induction over the segments $1 \ldots n$ of the synthesis trace:
  \begin{itemize}
    \item \textit{Base case}: $ \kappa_0 \ \mathsf{wf}_{H}$.
        This case is trivial, as each value in $H_0$ is initialized to the target scope~$s_d$, which is a composite path by \textsc{CP-Single}.
    \item \textit{Inductive case}: Follows immediately from~\cref{lem:expand-eventually-preserves-wfh}.
  \end{itemize}
  This proves $\kappa_n \ \mathsf{wf}_{H}$, and hence each synthesized reference corresponds to a composite path.
\end{proof}

}{}

\end{document}